\documentclass[
    aps,prl,twocolumn,
    superscriptaddress]{revtex4-2}
    \usepackage{standalone}
    \usepackage{tikz,pgfplots}
    \usepackage{xcolor}
    \usepackage{hyperref}
    \usepackage{enumitem}
    \usepackage{mathtools}
    \usepackage{amsthm}
    \usepackage{amssymb}
    \usepackage{bm}
    \theoremstyle{plain}
    \newtheorem{theorem}{Theorem}
    
    \newtheorem{proposition}[theorem]{Proposition}

    \theoremstyle{definition}
    \newtheorem{definition}{Definition}[section]

    \theoremstyle{remark}

    \usepackage{physics}
    \usepackage{xcolor}
    \usepackage{quantikz}

    \graphicspath{{_assets/Plot/}}
    \DeclareGraphicsExtensions{.png,.pdf}

    \usetikzlibrary{pgfplots.dateplot}

    \newcommand{\vertiii}[1]{{\left\vert\kern-0.25ex\left\vert\kern-0.25ex\left\vert #1
    \right\vert\kern-0.25ex\right\vert\kern-0.25ex\right\vert}}
    \newcommand{\hilb}{\mathcal{H}}
    \newcommand{\spX}{\mathcal{H}}
    \newcommand{\spY}{\mathcal{Y}}
    \newcommand{\UCh}{\mathcal{U}}

    \newcommand{\id}{\mathbb{I}}
    \newcommand{\complex}{\mathbb{C}}
    \newcommand{\trans}[1]{#1^\mathsf{T}}

    \newcommand{\btheta}{\vb*{\theta}}
    \newcommand{\chst}{C_{\mathrm{HST}}}
    \newcommand{\csep}{C_D}

\DeclarePairedDelimiter\ceil{\lceil}{\rceil}
\DeclarePairedDelimiter\floor{\lfloor}{\rfloor}

    \newcommand{\appone}{Measurements and State Preparation}
    \newcommand{\apptwo}{Parameter Shift Rule}

\begin{document}
\title{Variational Quantum Circuit Decoupling}
\author{Ximing Wang}
\email{canoming.sktt@gmail.com}
\affiliation{
Nanyang Quantum Hub, School of Physical and Mathematical Sciences, Nanyang Technological University, Singapore
}
\author{Chengran Yang}
\email{yangchengran92@gmail.com}
\affiliation{
Centre for Quantum Technologies, National University of Singapore, Singapore
}
\author{Mile Gu}
\email{mgu@quantumcomplexity.org}	
\affiliation{
Nanyang Quantum Hub, School of Physical and Mathematical Sciences, Nanyang Technological University, Singapore
}
\affiliation{
Centre for Quantum Technologies, National University of Singapore, Singapore
}

\affiliation{MajuLab, CNRS-UNS-NUS-NTU International Joint Research Unit, UMI 3654, Singapore 117543, Singapore}
\date{\today}

\begin{abstract}
    Decoupling systems into independently evolving components has a long history of simplifying seemingly complex systems. They enable a better understanding of the underlying dynamics and causal structures while providing more efficient means to simulate such processes on a computer. Here we outline a variational decoupling algorithm for decoupling unitary quantum dynamics -- allowing us to decompose a given $n$-qubit unitary gate into multiple independently evolving sub-components. We apply this approach to quantum circuit synthesis - the task of discovering quantum circuit implementations of target unitary dynamics. Our numerical studies illustrate significant benefits, showing that variational decoupling enables us to synthesize general $2$ and $4$-qubit gates to fidelity that conventional variational circuits cannot reach.
\end{abstract}

\maketitle

When studying a system of coupled harmonic oscillators, a key approach is to decouple its dynamics into various normal modes. We can then study each mode individually, enabling a deeper understanding of underlying dynamics. In the era of computer simulation, this divide-and-conquer approach has further operational importance, enabling leverage of parallel processing, and reducing the computational costs of simulation or optimization~\cite{pingali2011,trinder1998}. Decoupling complex systems into simpler components has seen success in diverse settings, from modeling the motion of human hands to hydraulic simulation~\cite{wu1999,wang2018}.

    \begin{figure} [t]
        \centering
        \includegraphics[width=\columnwidth]{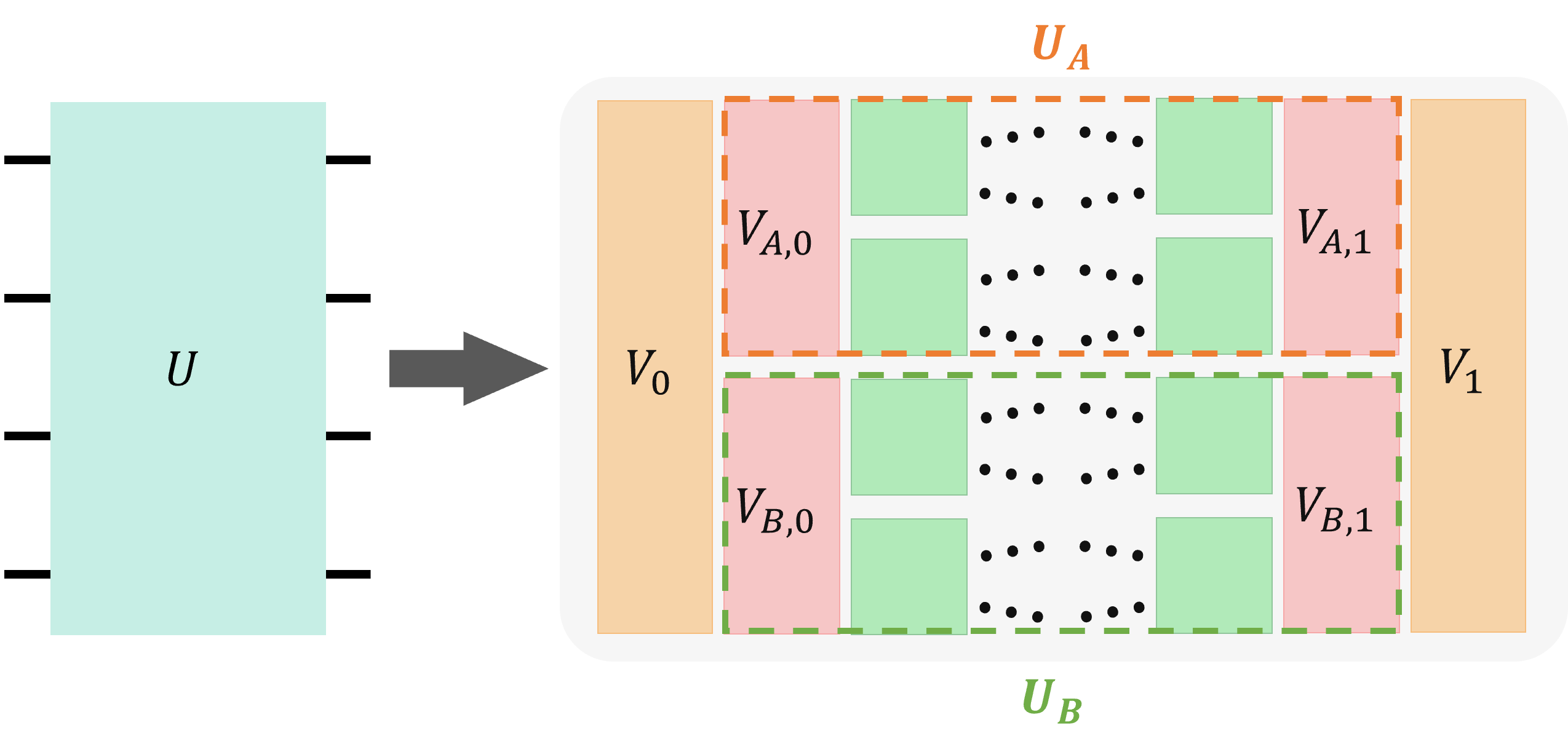}
        \caption{\textbf{Variational Quantum Circuit Decoupling.} Given a target complex unitary quantum operation $U$, the decoupling of $U$ involves identifying shallow pre and post-processing operation $V_0$ and $V_1$, such that $U \approx V_1(U_A\otimes U_B)V_0$, for some operators $U_A$ and $U_B$ act locally on subsystems $A$ and $B$ respectively.
        This procedure can be repeated recursively, breaking down $U_A$ and $U_B$ into smaller components, until all subsystems are sufficiently small.
        Take a $4$-qubit system for example, the target unitary operator $U$ is first decomposed by $V_0$ and $V_1$ into local operators on subsystems $A,B$, and the local operator $U_A$ ($U_B$) is further decomposed by $V_{A,0}$ and $V_{A,1}$ ($V_{B,0}$ and $V_{B,1}$) into single qubit unitary operators.
        }
        \label{fig:task}
    \end{figure}

\emph{Can decoupling also help us in the discovery of simpler quantum circuits}? Quantum computing promises immense computational benefits - from solving classically intractable problems to efficiently simulating quantum interactions. Yet, such tasks typically require the synthesis of complex $n$-qubit unitary operations using elementary quantum gates. While systematic decomposition techniques exist, all such generic solutions require gate counts that scale exponentially with $n$. To reduce such intractable realizations into a shallow circuit - even when possible - remains a highly intractable (NP-Complete) task~\cite{botea2018,kusyk2021}. Thus, quantum circuit compilation - the capacity to design a simpler circuit for implementing a given unitary process - is a pressing task; motivated by recent variation algorithms for quantum circuit discovery ~\cite{khatri2019,nemkov2023}.

Here, we proposed variation circuit decoupling -  a divide-and-conquer approach that enhances such existing variational circuit algorithms~\cite{khatri2019,nemkov2023}. The procedure involves designing a variational decoupling algorithm that first breaks a many-qubit interaction into smaller components (see figure~\ref{fig:task}). We can then optimize these smaller components individually or break them down further via recursive applications of variational decoupling. We demonstrate the method for compiling circuits approximating arbitrary two and four-qubit gates, where it identifies circuits of much higher fidelities than direct variational methods. 

\noindent\textbf{Framework --}\label{sec:decoupling} Here, we adopt the premise of variational quantum circuit discovery, where we are given black-box access to some unknown $n$-qubit unitary quantum process $U$~\cite{khatri2019,nemkov2023}. This could represent a complex quantum circuit we wish to simplify; some unknown quantum device we wish to reverse-engineer, or some natural quantum process we wish to simulate. Our goal is to approximate this process with a sequence of elementary quantum gates - subject perhaps to constraints on circuit complexity or depth. 

Systematic solutions to this problem are clearly intractable. Process tomography to determine the matrix elements of $U$ alone would scale exponentially with $n$, as would any general method of decomposing $U$ into elementary gates. Quantum-aided variational approaches aid in accelerating the process by use of quantum computers themselves~\cite{khatri2019,nemkov2023}. They typically involve two components: (a) An ansatz, a family of quantum circuits $W(\btheta)$ parameterized by some $k$-dimensional vector $\btheta = (\theta_0,\dots,\theta_{k-1})$, and (b) A cost function $C$ that that measures the performance loss of each candidate circuit in approximating $W$. The goal is to identify the parameter set $\btheta$ so that the corresponding circuit $W(\btheta)$ has a minimal cost. $\btheta$ corresponds to easily adjustable parameters of some proposed circuit architectures (e.g. magnitude of rotations on various single-qubit gates), such that knowledge of $\btheta$ allows the physical synthesis of the associated quantum circuit $W(\vb*{\theta})$. Meanwhile, the cost function is generally chosen to be normalised (takes a maximum value of $1$) with the following properties:
\begin{itemize}
    \item \textbf{Measurable and Gradient Measurable}, such that \(C(W)\) and each partial derivative \(\pdv{\theta_i}C\)
    can be efficiently estimated. i.e., they each can be determined to target precision $\varepsilon$ using at most $\text{poly}(n,1/\varepsilon)$ calls to $W$ and classical computing costs that scale efficiently with $n$.
    \item \textbf{Faithful}, such that \(C\) is non-negative, and \(C(W) = 0\) if and only if no general $n$-qubit quantum circuit are preferred over \(W\).
\end{itemize}
For example, $1-\bar{F}(\hat{U}, U)$ is a natural candidate for the cost function~\cite{khatri2019}, where $\bar{F}(\hat{U}, U)$ is the gate fidelity of the synthesize unitary $\hat{U}$ with respect to the target unitary $U$. This then allows gradient descent-based methods to find $\btheta$ within the parameterized circuit to minimize the cost (e.g. Adam~\cite{kingma2017}). 
However, the rate at which variational circuits converge can become severely limited due to barren plateaus - especially in scenarios where circuits being optimized over had no constraints \cite{mcclean2018,wang2021a,cerezo2021a}.

\textbf{Variational Decoupling} -- We adopt a divide-and-conquer approach. Instead of attempting to learn a circuit decomposition for the entirety of $U$, we break the circuit down into a series of smaller circuits as illustrated in figure~\ref{fig:task}. Specifically we first divide $n$-qubit system into two subsystems $A$ and $B$ of approximately equal size (i.e., of size $\floor{n/2}$ and $\ceil{n/2}$). We then identify pre-processing unitary $V_0$ and post-processing unitary $V_1$ such that  $V_1^\dagger UV_{0}^\dagger$ is most closely approximated by local unitary evolution on $A$ and $B$. That is, identify 
$V_0$ and $V_1$ such that there exists $U_A$ on $A$ and $U_B$ on $B$ in which $V_{1}^\dagger UV_{0}^\dagger$ most closely approximates $U_A \otimes U_B$.

We document below a variational means to learn gate decomposition for $V_0$ and $V_1$ without needing any knowledge of unitary operators $U_A$ and $U_B$. After which, we may repeat the decoupling method on each of the localized unitary gates, breaking each of them down into two smaller gates on $\approx n/4$ qubits. This can proceed until all localized operators are sufficiently small for conventional variational methods (e.g. a single qubit). Finally, various sub-problems are recombined to give a full circuit decomposition of $U$.

The key behind the variational decoupling algorithm is to find an appropriate cost function. First, let $V_0(\vb*{\theta}_0)$ and $V_1(\vb*{\theta}_1)$ be the parametrizations of the circuits $V_0$ and $V_1$ with respective parameters $\btheta_1$ and $\btheta_2$. Meanwhile set $W = V_{1}^\dagger UV_{0}^\dagger$ as the resulting quantum process formed by pre and postprocessing of $U$. Since we can always engineer $W$ given $U$, $V_0$, and $V_1$, we can define our cost function concerning $W$ alone for convenience. Here we desire a normalized cost function that (i) is faithful, such that $C(W) = 0$ when $U = V_{1} UV_{0}^{\dag} = U_A \otimes U_B$ or $(U_A \otimes U_B) S$, and (ii) exhibits both self and gradient measurability. Here we propose the cost function as the decoupling cost $\csep$ defined as
 \begin{equation} \nonumber 
        \begin{aligned}
        \csep(W)
        &= \iint D[
            W(\ketbra{\psi}_{A}\otimes\ketbra{\phi}_{B})W^\dagger
        ] \dd \ket{\psi} \dd \ket{\phi},
        \end{aligned}
 \end{equation}
where $D(\rho_{A, B}) = \frac{4^{m}}{(4^{m}-1)}\frac{L(\rho_A) + L(\rho_B)}{2}$, $L(\rho) = 1-\Tr[\rho^2]$ is the linear entropy and the integration is take over Haar random initial states $\ket{\psi}_{A}$ on $A$ and $\ket{\Phi}_{B}$ on $B$.
Here $m$ is either the number of qubits in $A$ or $B$, whichever is smaller.
The normalization factor $\frac{4^{m}}{4^{m}-1}$ is chosen such that $D(\rho_{A,B})$ ranges from $0$ to $1$ for all pure state $\rho_{A,B}$.
We now illustrate the following:

\begin{theorem}
        The decoupling cost $\csep$ is efficiently measurable, efficiently differentiable, and faithful.
\end{theorem}

We illustrate in figure~\ref{fig:Procedure} (c) that $\csep$ is \emph{efficiently measurable}. Specifically, the circuit involves $2n$ qubits, a constant-depth overhead from implementing $V_1^\dagger UV_{0}^\dagger$ Note that while $C_D$ is defined using Haar random initial states, measurement of $C_D$ requires only preparation of random states in the computational basis (see supplementary material~\emph{\appone} for proof).

We establish the \emph{efficiently differentiability} of $\csep$ by generalising the parameter-shift rule -- a noise-resilient, hardware efficient method, and provides exact gradients through quantum measurement~\cite{mitarai2018,schuld2019}.
As the parameterized circuit $W$ appears twice in evaluating the decoupling cost $\csep$, so does each parameterized gate in $W$, the conventional parameter-shift rule cannot be applied directly.
A parameter-shift rule addressing multiple occurrences of parameters needs to be considered~\cite{li2017,lu2017}.
We use $W_\alpha$ and $W_\beta$ to specify each of the two copies of $W$ used to evaluate $\csep(W) = \csep(W_\alpha, W_\beta)$ as in figure~\ref{fig:Procedure}.
Given the ansatzes whose parameterized gates are only Pauli rotations, we show that the gradient $\pdv{\theta_i} \csep$ for each $\theta_i$ in $\vb*{\theta}$ can be written down as
\begin{equation}\label{eq:gradient}
    \begin{aligned}
    \pdv{\theta_i} \csep =&
    \frac{1}{2} \csep(W_\alpha^+,W_\beta) - \frac{1}{2} \csep(W_\alpha^-,W_\beta)\\
    &+ \frac{1}{2} \csep(W_\alpha,W_\beta^+) - \frac{1}{2} \csep(W_\alpha,W_\beta^-) \ .
    \end{aligned}
\end{equation}
where $W^+_i$ and $W^-_i$ are the ansatz with the parameter $\theta_i$ in $W(\vb*{\theta})$ updated to $\theta_i+\frac{\pi}{2}$ and $\theta_i-\frac{\pi}{2}$ respectively (see supplementary materials Sec. III).
Each term in~\eqref{eq:gradient} can then be directly evaluated by the circuit in figure~\ref{fig:Procedure}, leading to the efficient measurement of the gradient.

To establish \emph{faithfulness}, first observe that the linear entropy is used to quantify the entanglement between two subsystems~\cite{santos2000}. In particular, the non-negativity of the linear entropy guarantees the non-negativity of the decoupling cost $\csep$. Meanwhile, $D(W (\rho_A\otimes \rho_B)W^\dagger)$ quantifies the entanglement $W$ generates between $A$ and $B$ when applied to a product state $\rho_{AB} = \rho_A\otimes \rho_B$. Consequently, our cost function $\csep$ represents the expected entanglement generated by $W$ when averaged across all pure product inputs $\ket{\psi}_A \otimes \ket{\phi}_B$. As such, $C_D(W) = 0$ if and only if the operator $W$ is equivalent to the identity or the swap gate up to local unitary gates. In addition, in the supplementary material Sec.IV, we prove that $C_D$ constrains the ultimate gate fidelity we can achieve after decoupling:
\begin{theorem}
    For a given $n$-qubit unitary operator $W$, the gate fidelity $\bar{F}(U_A\otimes U_B,W)$ for all $U_A$ and $U_B$ is bounded by the decoupling cost $\csep$, such that
    \begin{equation}
        \bar{F}_{\max}^2 \leq  \mathrm{min} (1- \csep + \frac{3}{2^n+1}, 1) \ .
    \end{equation}
\end{theorem}

We can then proceed with variational decoupling by choosing appropriate ansatzes for $V_0(\vb*{\theta}_0)$ and $V_1(\vb*{\theta}_1)$; such that $W$ is parameterized by the joint vector $\vb*{\theta}^* := \vb*{\theta}_0^*\oplus \vb*{\theta}_1^*$. For each candidate $V_0(\vb*{\theta}_0)$ and $V_1(\vb*{\theta}_1)$, we can construct a realization of $W$, and use above methods to efficiently measure $\csep$ and its various gradients, and thus employ gradient-descent-based optimization method to find suitable $V_0$ and $V_1$ that decouples $U$. Note that, just as in standard circuit compilation, we may also choose to limit the gate depth of $V_0$ and $V_1$ should we wish for shallow circuit approximations of $U$.

    \begin{figure}[htp]
    \centering
        \includegraphics[width=\columnwidth]{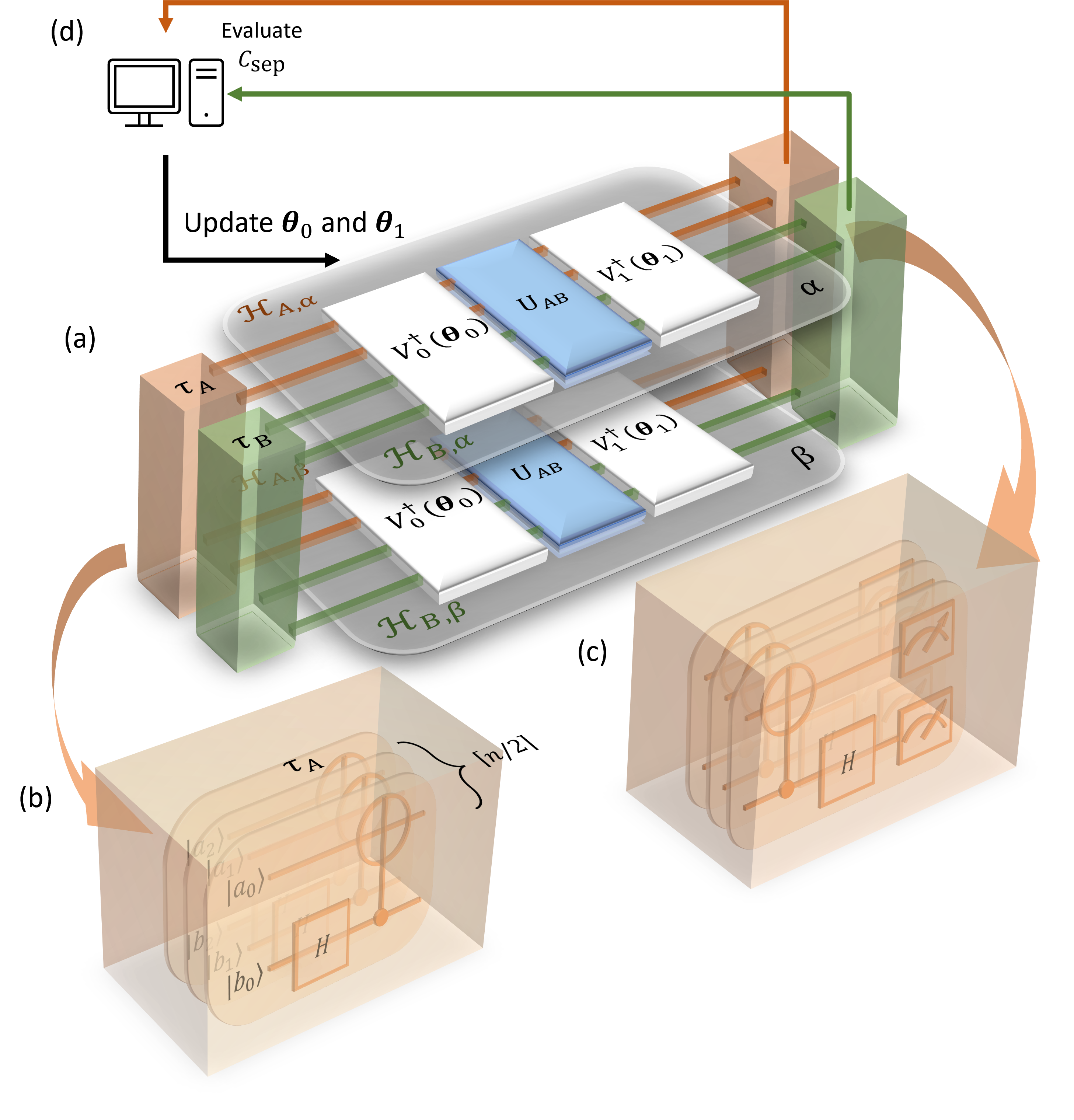}
        \caption{\emph{The variational decoupling algorithm.} To decouple an $n$-qubit circuit, our algorithm begins with two systems of $n/2$-qubits - denoted as the upper and lower layer in (a), which we label $\alpha$ and $\beta$. Each layer is further divided into the $2$ subsystems we wish to decouple, resulting in the $4$ subsystems $\spX_{A,\alpha}$, $\spX_{A,\beta}$ of dimensions $d_A$ and $\spX_{B,\alpha}$, $\spX_{B,\beta}$ of dimensions $d_B$. We apply $W= V_1^\dagger U V_0^\dagger$ to each of subsystem $\alpha$ and $\beta$.
        The input state is initialized as $\tau_A\otimes \tau_B$, where $\tau_{k} = \frac{2P_s^{(k)}}{d_{k}^2+d_{k}}$, and  $P^{(k)}_s$ is the projector to the symmetric subspace of $\spX_{k,\alpha} \otimes \spX_{k,\beta}$ for $k \in \{A,B\}$. This is done by first preparing each qubit in $\spX_{k,\alpha}$ and $\spX_{k,\beta}$ in a random computational basis state, followed by qubit-wise application of the Hadamard gate for each $\beta$-layer qubit. Each qubit is then entangled with a corresponding qubit in $\spX_{k,\alpha}$ via a CNOT gate (b). Meanwhile, we can efficiently measure cost function $\csep$ and its gradient using a destructive swap test - which similarly involves only qubit-wise CNOT and Hadamard gates, followed by projective measurement (c). This information can then be used to tune $\theta_0$ and $\theta_1$ using techniques from variational circuits (d).
        }
        \label{fig:Procedure}
    \end{figure}

\noindent\textbf{Performance Comparisons --}
We demonstrate the improved efficacy of our decoupling method in two distinct scenarios:

\begin{itemize}
\item \emph{The exact compilation of a general $2$-qubit gate} - where the ansatz is shown figure~\ref{fig:2result}. Here, the circuit family is universal, and can theoretically synthesize any $2$-qubit unitary operator~\cite{shende2004}. Our decoupling algorithm breaks this into two stages. The first (coloured orange) decouples $U$ into two single-qubit circuits. The second stage (red) then optimises these circuits. Note that in this special case, we have set $V_1 = I$, since this is already sufficient for compiling an arbitrary $2$-qubit gate.

\item \emph{A shallow circuit approximation of an arbitrary $4$-qubit gate}, where the structure of the ansatz follows the figure~\ref{fig:training_decoup4qubit}. For both $V_0$  and $V_1$, we utilize four layers of decoupling circuit, where each layer comprises a shallow circuit and single qubit and control gates. The circuit compilation is done via three stages. Stage one (orange) decouples the circuit into two generic $2$-qubit gates, $V_A$ and $V_B$. The second stage (red) decouples these 2-qubit gates further into $4$ single-qubit interactions. Finally, stage $3$ then optimizes each of the four single-qubit gates. 
\end{itemize}

In each setting, we select target $U$ at random according to the Haar measure. The former scenario represents the case where we have a universal circuit ansatz and which to exactly compile a unitary. The latter represents the scenario where we wish to use a shallow circuit with limited expressivity to approximate an arbitrary $4$-qubit gate. 

We benchmark our method against standard variational approaches, where all circuit parameters are optimized by gradient descent to maximize the gate fidelity $\bar{F}$ simultaneously.
In the standard approach, this optimization is done by minimizing a cost function $\chst$ or the localized cost function $C_{\mathrm{LHST}}$~\cite{khatri2019}.
Here $\chst = \frac{d+1}{d}(1-\bar{F})$ and $C_{\mathrm{LHST}} = \frac{1}{n}\sum_i C_{\mathrm{HST}}^{(i)}$ are faithful measures who reaches $0$ when $\bar{F}$ approaches $1$. The $C_{\mathrm{HST}}^{(i)}$ are the the cost function $\chst$ restricted to the $i$th qubit.
For a fair comparison, all methods use the same circuit parameterization and thus have identical expressivity and all are trained using the ADAM~\cite{kingma2017} method with identical hyper-parameters~\footnote{The $\alpha = 0.01$, $\beta_1 = 0.8$, $\beta_2 = 0.9$.}. The results are shown in figure~\ref{fig:2result} (2-qubit) and figure~\ref{fig:training_decoup4qubit} (4-qubit). In both cases, performance is gauged by how $\bar{F}(U,\hat{U})$ improves, i.e., the fall of $\chst$ with the number of training iterations. We see that:
\begin{itemize}
    \item In $2$-qubit gate decomposition, our decoupling approach converges a lot more quickly. For example, it achieves a fidelity of $0.9999$ using the same amount of time where conventional methods approach a fidelity of $0.999$
    \item In shallow circuit approximation of general $4$-qubit gates, our decoupling method reaches circuit fidelity around $0.7$, significantly ahead of conventional methods which achieves a fidelity of around $0.5$. Note that high fidelity values are not expected here as our ansatz - being shallow - has limited expressivity. In supplementary materials V, we also include an analysis of a third scenario where target $4$-qubit unitaries can be expressed by our shallow circuit. There our decoupling method continues can reach a circuit fidelity of $0.998$ vs $0.9$ in using conventional methods.
\end{itemize}

In all cases, conventional variational circuits exhibit steady gate fidelity gain till some saturation point. Our decoupling method, on the other hand, only minimizes $C_{\mathrm{LHST}}$ in the final stage. Previous stages concentrate on decoupling the unitary and do not immediately improve gate fidelity. However, once decoupled, convergence to a low $\chst$ occurs very quickly - more than compensating for the delayed start. Moreover, the final fidelity obtained shows marked improvement for both exact gate decomposition and shallow circuit approximation.

  \begin{figure}[htp]
      \centering
      \includegraphics[width=0.43\textwidth]{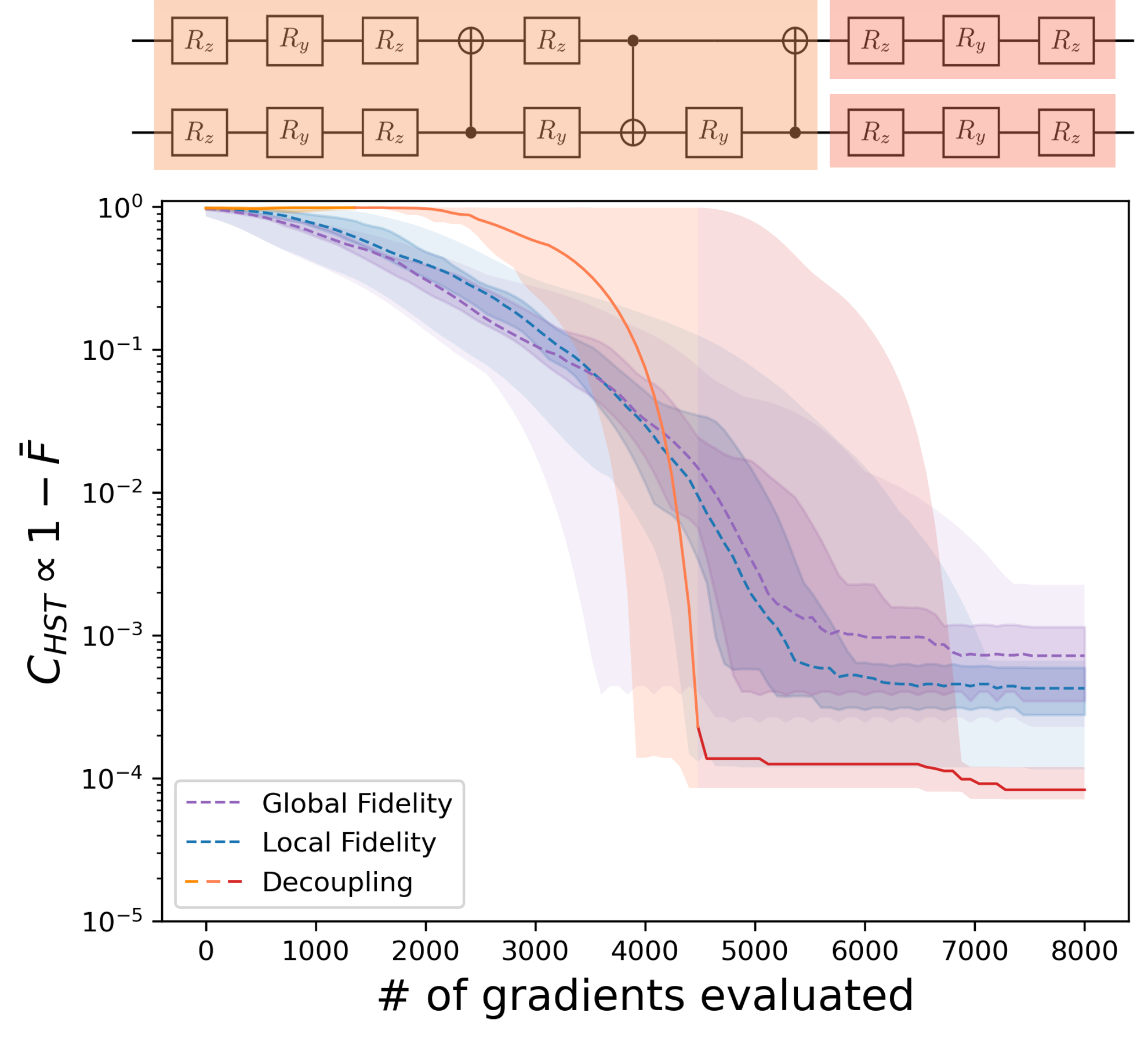}
      \caption{\emph{Variational Compilation of a general 2-qubit unitary with ADAM optimization.}
      The fidelity $\bar{F}$ is compared between the direct compilation (purple for global cost functions and blue for local cost functions) and our decoupling method (solid lines).
      For each method, the training is repeated $20$ times with different initializations.
      Our method divides the task into two steps, each represented by a different color.
      The first phase (yellow) decouples the unitary to single-qubit unitary operators, and the second phase (red) maximizes the fidelities.
      The phase transition region is colored orange, where some of the processes are in phase one and the rest are in phase two.
      The data of each method are divided into quartiles, and the lines represent the median of each process.
      The second and third quartiles are shaded with the corresponding color.
      The preprocessing ansatz is chosen to be the universal two-qubit circuit~\cite{shende2004}, and the post-processing is set to identity. We see that the average fidelity does not increase during the decoupling stage, but greatly enhances training in the second stage by breaking up the problem into more tractable components. The resulting compiled circuits have significantly higher fidelity, such that $1 - \bar{F}$ is reduced by a factor of $3$ over the state of the art.
      }
  \label{fig:2result}
  \end{figure}

  \begin{figure}[htp]
      \centering
      \includegraphics[width=0.45\textwidth]{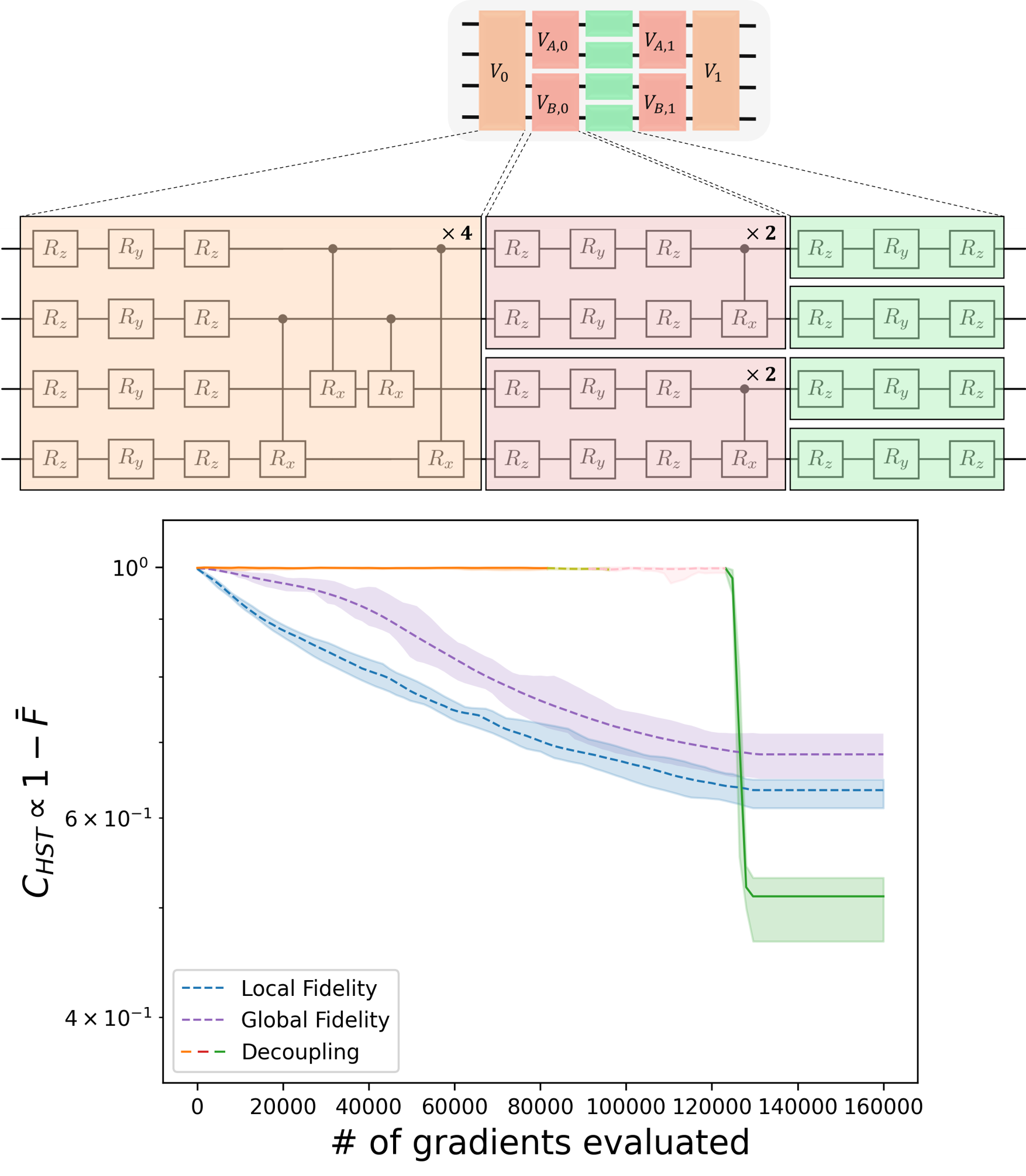}
      \caption{\emph{Variational Compilation of generic 4-qubit unitary with ADAM optimization.}
      Similar to the 2-qubit case, the data of each process is collected from $20$ independent training processes, while the regions between the first and last quartiles are shaded with the corresponding colors as in FIG~\ref{fig:2result}.
      Here, the optimisation is divided into three phases.
      The circuits are divided into 3 types: 4-qubit ansatzes (orange), 2-qubit ansatzes (red), and the middle single qubit gates (green).
      The layers are repeated for $4$ times for $V_0,V_1$, and the layers in the 2-qubit ansatzes $V_{A,0},V_{B,0},V_{A,1},V_{B,1}$ repeat twice.
      The ansatzs are designed to explore stronger expressive power with a limited number of gates, which can be generalized to any number of qubits~\cite{du2020}.
      In this case, our ansatz -- being low depth -- is non-universal. Nevertheless, we see that variation decoupling continues to achieve circuits of higher fidelity.
      }
      \label{fig:training_decoup4qubit}
  \end{figure}

    \textbf{Discussion} -- Here, we outline a variational means to decouple complex quantum circuits.
    This is enabled by the proposal of an efficiently measurable cost function $\csep$, that allows us to identify suitable pre- and post-processing circuits to decouple a $2n$-qubit unitary to a tensor product of two $n$-qubit unitary operators.
    Through the recursive application, we illustrated how variational decoupling enhances existing approaches to circuit compilation -- a key task in the design of quantum circuits in both near-term of future universal quantum computers. In benchmarks, variational decoupling enabled us to compile
    circuits realization of arbitrary two and four-qubit quantum gates to significantly higher fidelity. By setting the allowed depths of pre-processing and post-processing circuits, our decoupling method can be adapted to discovering circuits of varying expressivity, ranging from very shallow circuits for NISQ devices to deep circuits that represent general unitary operations.

    There are many promising avenues for future research. One is to consider hybrid digital-analog quantum computing paradigms~\cite{lamata2018}, where pre or post-processing blocks may include analog evolution that provides efficient means to generating ansatz of increased expressivity. Meanwhile, the decoupling of complex interactions has motivations beyond computational advantage. In classical analogs such as coupled oscillators, identification of normal modes helped better express underlying dynamics - and thus benefits intrinsic human understanding. We anticipate a similar advantage in the quantum regime. One could, for example, envision a seemingly complex many-body quantum system evolving under some joint Hamiltonian $H$. By setting $V_0 = V^\dag_1$, our methodologies could also help identify potential sub-spaces where there is minimal coupling - and thus aid the discovery of underlying causal restructure. Finally, the minimal mutual information between arbitrary bipartitions of a many-body system is considered a fundamental measure of indivisibility -- and is remarkably a lot smaller for general quantum systems due to the potential to decouple in an entangling basis ~\cite{tegmark2015}. It would certainly be interesting to understand if this phenomenon indicates that decoupling is particularly beneficial in the quantum regime.

\begin{acknowledgments}
\textit{Acknowledgments} ---  This work is supported by the National Research Foundation, Singapore, and Agency for Science, Technology and Research (A*STAR) under its QEP2.0 programme (NRF2021-QEP2-02-P06), the Singapore Ministry of Education Tier 1 Grant RG77/22 (S) and the FQXi R-710-000-146-720
Grant “Are quantum agents more energetically efficient
at making predictions?” from the Foundational Questions Institute and Fetzer Franklin Fund (a donor-advised fund of the Silicon Valley Community Foundation).
\end{acknowledgments}

\onecolumngrid
\appendix
\newpage
{\hfill\centering \LARGE\textbf{Supplementary Material}\hfill}
\vspace{6em}

\section{Preliminary - Quantum Channels}

In the quantum theory, the states of a quantum system are described by the density operators $\rho$ that act on a Hilbert space $\spX$, such that the operators $\rho$ are positive semidefinite and $\tr[\rho]=1$.
A linear map $\Gamma$ from a system $\spX$ to system $\spY$ is called a \emph{quantum channel} if it maps a density operator in $\spX$ to a density operator in $\spY$.
Such property of the linear maps is usually broken into two constraints: 1. the map is completely positive (CP), and 2. the map is trace-preserving (TP).
A map $\Gamma:\spX\to \spY$ is CP if, for all ancillary system $\mathcal{A}$, $\Gamma\otimes \mathbb{I}_A$ maps positive semidefinite operators $\rho\in \spX\otimes \mathcal{A}$ to positive semidefinite operators $\Gamma\otimes\mathbb{I}_A(\rho)$.
A map is TP if it preserves the trace of operators, such that $\tr[\Gamma(\rho)] = \tr[\rho]$ for all operators defined on $\spX$.
For convenience, we usually investigate these two constraints with the matrix representations of the linear maps.
Choi-Jamiolkowski operator~\cite{jamiolkowski1972,watrous2018} is one of the most commonly used representations of quantum channels.
For each quantum channel $\Gamma$, its Choi-Jamiolkowski operator $J^\Gamma$ is unique.
\begin{definition}
    The Choi-Jamiolkowski operator $J^\Gamma$ of a linear map $\Gamma: \spX\to \spY$ is defined as
    \begin{equation}
        J^\Gamma = \left(\Gamma\otimes \id_S\right) \ketbra{\Psi^+} \ ,
    \end{equation}
    where $\ket{\Psi^+} = \sum_i \ket{i}_\spX\ket{i}_{\spX^\prime}$, $\spX^\prime$ is a copy of $\spX$, and $\ket{i}$ forms a basis of $\spX$.
\end{definition}
The properties of CP and TP can be characterized easily with Choi-Jamilkowski operators~\cite{watrous2018}:
\begin{theorem}
    A linear map $\Gamma$ is a quantum channel if and only if $J^\Gamma$ is positive semidefinite (CP), and $\tr_\spY[J^\Gamma] = \mathbb{I}_\spX$ (TP).
\end{theorem}
Choi-Jamiolkowski operator also provides a convenient way to represent the input states and the measurements on the outputs of the quantum channels~\cite{chiribella2008,chiribella2009}:
\begin{theorem}\label{thm:Choi}
    Given quantum state $\rho$ defined on $\spX$, an observable $O$ defined on $\spY$, and a quantum channel $\Gamma: \spX\to \spY$, the expectation value of $O$ with respect to $\Gamma(\rho)$ is given by
    \begin{equation}
        \tr[\Gamma(\rho)O] = \tr[\trans{(\rho\otimes O)} J^\Gamma] \ .
    \end{equation}
\end{theorem}

In the scope of this paper, we will focus on the endomorphisms of a system $\spX$, which are those linear maps $\Gamma: \spX\to \spX$ whose output system is identical to the input system.
More specifically, we will focus on the quantum channels that are induced by unitary operators, namely unitary channels.
Let $\UCh$ denote the unitary channel with respect to a unitary operator $U$ acts on $\spX$, such that $\UCh(\cdot) = U (\cdot) U^\dagger$.
Then the Choi-Jamiolkowski operator of $\UCh$ is given by
\begin{equation}
    J^\UCh = \qty(U \otimes \id_\spX) \ketbra{\Psi^+} \qty(U^\dagger\otimes \id_\spX) \ ,
\end{equation}
which is a rank-one projector.
For a given unitary $U$, there is an unique $\ket{\psi}$ (up to a global phase), such that $J^\UCh = \ketbra{\psi}$.
In the special case that $U = \mathbb{I}$, the state is the unnormalized maximally entangled state $\ket{\Psi^+}$.
According to the TP property of quantum channels, $\ketbra{\psi}{\psi} = \tr[J^\UCh] = \tr[\mathbb{I}_\spX] = d$, where $d$ is the dimension of $\spX$.

\section{\appone} 
\label{sec:SwapTest}

It is known that the purity, $\Tr[\rho^2]$, of a state $\rho$ can be measured with swap test efficiently~\cite{barenco1997}.
With the ability to perform a controlled swap between two copies of $\rho$, the value of $\Tr[\rho^2]$ is proportional to the probability of measuring $0$ in the circuit of figure~\ref{fig:swap_test}.

\begin{figure}[htp]
\centering
\begin{quantikz}
    \lstick{$\ket{0}$} & \gate{H}  & \ctrl{2}  &  \gate{H} & \meter{} \\
    \lstick{$\rho$}    & \qw       & \gate[swap]{} &  \qw    & \qw   \\
    \lstick{$\rho$}    & \qw       &           &  \qw      &  \qw  \\
\end{quantikz}
\caption{Conventional swap test.}\label{fig:swap_test}
\end{figure}
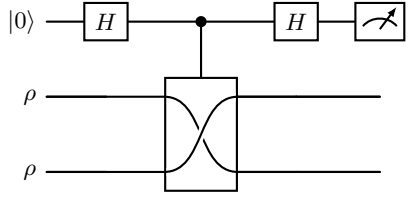

\begin{figure}[htp]
\centering
\begin{quantikz}
    \lstick{$\rho$}   &  \targ{}    &  \qw      & \meter{z_1}  \\
    \lstick{$\rho$}   &  \ctrl{-1}  &  \gate{H} & \meter{z_2}  \\
\end{quantikz}
\caption{Destructive swap test}\label{fig:destructive_swap_test}
\end{figure}
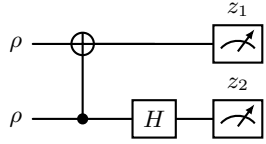

Alternatively, the measurement of purity can be done with pair-wise Bell measurements instead of the controlled swap test~\cite{garcia-escartin2013, larose2019}.
This is done by a bit-wise bell measurement on a pair of systems as shown in figure~\ref{fig:destructive_swap_test}, which are called \emph{destructive swap test}.
For large quantum systems, the destructive swap test is more practical than the conventional swap test.
In the conventional scheme, the control qubit must maintain a strong quantum correlation with all qubits that are been swapped.
However, in the destructive swap test, each pair of qubits is measured separately, and the computation of purity only requires classical correlations between the measurement outcomes.
In summary,
\begin{theorem}\label{thm:measurement}
    The purity of a state $\rho$ can be measured as an observable $S$, where $\Tr[\rho^2] = \Tr[\rho\otimes \rho S]$, such that the evaluation can be measured with a constant depth quantum circuit and a linear time classical post-processing.
\end{theorem}

\begin{proof}
The proof starts from the single-qubit case, where $\rho$ acts on
$\mathbb{C}^2$.
Let $\ket{\Psi^+}:= \frac{1}{\sqrt{2}} \qty(\ket{00}+\ket{11})$, the
projective measurement in circuit~\ref{fig:destructive_swap_test} is given
by the projectors that project to
\begin{equation}
    \ket{\Psi_{z_1,z_2}} := \sigma_X^{z_1}\otimes \sigma_Z^{z_2} \ket{\Psi^+} \ ,
\end{equation}
where $a$ and $b$ take values from $\qty{0,1}$, which are the result of the measurement on each subsystem.
This is exactly the Bell basis, and hence the measurement~\ref{fig:destructive_swap_test} is commonly known as the Bell measurement.
The Bell basis can be rewritten as
\begin{equation}
    (-i)^{z_1\cdot z_2} \sigma_{z_1,z_2} \otimes \id \ket{\Psi^+} \ ,
\end{equation}
where $\sigma_{z_1,z_2} = i^{z_1,z_2}\sigma_X^{z_1} \sigma_Z^{z_2}$ are the Pauli operators:
$$\mathbb{I}=\sigma_{0,0};\quad\sigma_X=\sigma_{1,0};\quad\sigma_{Z}=\sigma_{0,1};\quad\sigma_{Y}=\sigma_{1,1} \ .$$
This alternative form is derived from the identity, where for all operators $A, B$ acting on $\complex^2$,
\begin{equation} \label{eq:Bell_identity}
    A\otimes B \ket{\Psi^+} = (A\trans{B})\otimes \id \ket{\Psi^+} \ .
\end{equation}
As the result, the probability of obtaining measurement outcome $(z_1,z_2)$ in the destructive swap test~\ref{fig:destructive_swap_test} is given by
\begin{equation}
    p(z_1,z_2) = \bra{\Psi^+} \qty(\sigma_{z_1,z_2} \rho \sigma_{z_1,z_2})\otimes \rho \ket{\Psi^+} \ ,
\end{equation}
which uses the fact that Pauli operators are Hermitian.
By applying \eqref{eq:Bell_identity} to the second copy of $\rho$, together with the definition of $\ket{\Psi^+}$,
\begin{equation}
    p(z_1,z_2) = \frac{1}{2}\Tr[\rho\sigma_{z_1,z_2}\trans{\rho}\sigma_{z_1,z_2}] \ .
\end{equation}

As the transpose is a Hermitian-preserving trace-preserving linear map on operators, it has a unique affine decomposition in terms of the Pauli operators:
\begin{equation}\label{eq:trans}
    \trans{\rho} = \frac{1}{2}
    \sum_{z_1,z_2} (-1)^{z_1\cdot z_2} \sigma_{z_1,z_2} \rho \sigma_{z_1,z_2} \ .
\end{equation}
Also, since transpose is an involution, we may interchange $\rho$ with
$\trans{\rho}$ in \eqref{eq:trans}.
As a result,
\begin{equation}
\begin{aligned}
    \Tr[\rho^2] &= \Tr[\rho \frac{1}{2}
    \sum_{z_1,z_2} (-1)^{z_1\cdot z_2} \sigma_{z_1,z_2} \trans{\rho} \sigma_{z_1,z_2}] \\
    &= \sum_{z_1,z_2} (-1)^{z_1\cdot z_2} \frac{1}{2}
    \Tr[\rho \sigma_{z_1,z_2} \trans{\rho} \sigma_{z_1,z_2}] \\
    &= \sum_{z_1,z_2} (-1)^{z_1\cdot z_2}p(z_1,z_2) \ .
\end{aligned}
\end{equation}
This indicates that the value of $\Tr[\rho^2]$ can be obtained with a simple classical post-processing, where after each measurement, $(-1)^{z_1\cdot z_2}$ is computed, whose mean is the value of $\Tr[\rho^2]$.
This procedure is equivalent to the evaluation of the expected value of the observable
\begin{equation}
    S_1 := \sum_{z_1,z_2}(-1)^{z_1\cdot z_2} \ketbra{\Psi_{z_1,z_2}} \ .
\end{equation}

This can be easily extended to multi-qubit systems, where the Bell measurements are performed pairwisely between the two copies of the same state.
In this scenario, the measurement results are given by the Boolean vectors $\pmb{z}_1$ and $\pmb{z}_2 \in \mathbb{B}^n$, where the $i$th element of each vector corresponds to the measurement outcome of the $i$th pair of qubits.
The multi-qubit Bell states are given by
\begin{equation}
    \ket{\Psi_{\pmb{z}_1,\pmb{z}_2}} := \bigotimes_{i} \ket{\Psi_{z_{i,1},z_{i,2}}} \ ,
\end{equation}
where $\ket{\psi_{z_{i,1},z_{i,2}}}$ are the Bell states on $i$th qubit.
The probability of each measurement outcome is given by
\begin{equation}
    p(\pmb{z}_1,\pmb{z}_2) = \bra{\Psi^+}
    \qty(\sigma_{\pmb{z}_1,\pmb{z}_2} \rho \sigma_{\pmb{z}_1,\pmb{z}_2})
    \otimes \rho \ket{\Psi^+} \ ,
\end{equation}
where $\sigma_{\pmb{z}_1,\pmb{z}_2} := \bigotimes_{i} \sigma_{z_{i,1},z_{i,2}}$.
By applying identity \eqref{eq:Bell_identity} to all qubits simultaneously,
\begin{equation}
    p(\pmb{z}_1,\pmb{z}_2)
    = \frac{1}{2^n}\Tr[\rho\sigma_{\pmb{z}_1,\pmb{z}_2}
    \trans{\rho}\sigma_{\pmb{z}_1,\pmb{z}_2}] \ .
\end{equation}
Since the transpose of a $n$-qubit state is equivalent to partial transpose
each qubit,
\begin{equation}
\begin{aligned}
    \trans{\rho} &= \sum_{\pmb{z}_1,\pmb{z}_2} \qty(\prod_i \frac{1}{2}
    (-1)^{z_{i,1}\cdot z_{i,2}})
    \sigma_{\pmb{z}_1,\pmb{z}_2} \rho \sigma_{\pmb{z}_1,\pmb{z}_2} \\
    &= \frac{1}{2^n} \sum_{\pmb{z}_1,\pmb{z}_2} (-1)^{\pmb{z}_1\cdot \pmb{z}_2}
    \sigma_{\pmb{z}_1,\pmb{z}_2} \rho \sigma_{\pmb{z}_1,\pmb{z}_2} \ ,
\end{aligned}
\end{equation}
where $\pmb{z}_1\cdot\pmb{z}_2 := \sum_i z_{i_1}\cdot z_{i_2}$.
The purity can be expressed as
\begin{equation}
    \Tr[\rho^2] = \sum_{\pmb{z}_1,\pmb{z}_2} (-1)^{\pmb{z}_1\cdot \pmb{z}_2}
    p(\pmb{z}_1,\pmb{z}_2) \ .
\end{equation}
That is, to measure the purity, one Bell measurements~\ref{fig:destructive_swap_test} is applied to each pair of qubits.
After the outcomes are obtained, the value $(-1)^{\pmb{z}_1\cdot \pmb{z}_2}$ is computed, whose mean is the value of $\Tr[\rho^2]$.
As the Bell measurements are independent, all of them can be done simultaneously within two layers of quantum gates.
The post-processing is a simple parity checking, whose complexity grows linearly with the number of qubits.
The corresponding observables are
\begin{equation}
    S_n := \sum_{\pmb{z}_1,\pmb{z}_2}(-1)^{\pmb{z}_1\cdot \pmb{z}_2} \ketbra{\Psi_{\pmb{z}_1,\pmb{z}_2}} \ .
\end{equation}
\end{proof}

\begin{theorem}
    The observable $S$ in theorem~\ref{thm:measurement} is the swap operator between the two copies of $\rho$.
\end{theorem}
\begin{proof}
To check that $S_n$ is indeed the swap operator between the two copies of $\rho$, we first notice that $S_n$ can be expressed as two terms:
\begin{equation}
    S_n := \sum_{\pmb{z}_1\cdot \pmb{z}_2 = 0} \ketbra{\Psi_{\pmb{z}_1,\pmb{z}_2}}
    - \sum_{\pmb{z}_1\cdot \pmb{z}_2 = 1} \ketbra{\Psi_{\pmb{z}_1,\pmb{z}_2}}  \ .
\end{equation}
As Bell states $\ket{\Psi_{\pmb{z}_1,\pmb{z}_2}}$ are orthogonal to each other, the two terms are (unnormalised) orthogonal projectors.
As discussed in the proof of the theorem~\ref{thm:measurement}, the Bell basis can be expressed as
\begin{equation}
    \ket{\Psi_{\pmb{z}_1,\pmb{z}_2}} := (-i)^{\pmb{z}_1\cdot \pmb{z}_2}
    \sigma_{\pmb{z}_1,\pmb{z}_2} \otimes \mathbb{I} \ket{\Psi^+} \ .
\end{equation}
By swapping the two subspaces
\begin{equation}
\begin{aligned}
    S \ket{\Psi_{\pmb{z}_1,\pmb{z}_2}}
    &= (-i)^{\pmb{z}_1\cdot \pmb{z}_2} \mathbb{I} \otimes
    \sigma_{\pmb{z}_1,\pmb{z}_2} \ket{\Psi^+} \\
    &= (-i)^{\pmb{z}_1\cdot \pmb{z}_2}
    \trans{\sigma}_{\pmb{z}_1,\pmb{z}_2} \otimes
    \mathbb{I} \ket{\Psi^+} \\
    &= (-i)^{\pmb{z}_1\cdot \pmb{z}_2}(-1)^{\pmb{z}_1\cdot \pmb{z}_2}
    \sigma_{\pmb{z}_1,\pmb{z}_2} \otimes
    \mathbb{I} \ket{\Psi^+} \\
    &= (-1)^{\pmb{z}_1\cdot \pmb{z}_2} \ket{\Psi_{\pmb{z}_1,\pmb{z}_2}} \ .
\end{aligned}
\end{equation}
That is, $\ket{\Psi_{\pmb{z}_1,\pmb{z}_2}}$ is symmetric if $\pmb{z}_1\cdot\pmb{z}_2 = 0$, and is anti-symmetric if it is $1$.
By counting the dimension of the symmetric subspace, it is easy to show that those $\ket{\Psi_{\pmb{z}_1,\pmb{z}_2}}$ with $\pmb{z}_1\cdot\pmb{z}_2 = 0$ forms a complete basis of the symmetric subspace, and $\Pi_S := \sum_{\pmb{z}_1\cdot \pmb{z}_2 = 0} \ketbra{\Psi_{\pmb{z}_1,\pmb{z}_2}}$ is the projector to the symmetric subspace.
Similarly, those with $\pmb{z}_1\cdot\pmb{z}_2 = 1$ form a complete basis of the anti-symmetric subspace, and the second term, denoted $\Pi_A$, is the projector to the anti-symmetric subspace.
As a result, $S_n = \Pi_S - \Pi_A$ is the swap operator, which preserves the symmetric subspace and provides a $-1$ phase to the anti-symmetric subspace.
\end{proof}

\begin{theorem}\label{thm:state_prep}
    Given an observable $O$ defined for 2-copies of $\UCh(\rho) \in \hilb$, and consider the expectation value $E(\rho)$ with respect to a channel $\UCh$ that $E = \Tr[\UCh(\rho)\otimes \UCh(\rho) O]$.
    The Haar average of $E$, $\int \dd\psi E(\ketbra{\psi})$, can be evaluated by measuring $O$ on a single initial state $\tau\in \hilb^{\otimes2}$, such that $\tau$ can be prepared with a constant depth quantum circuit and a linear time classical pre-processing.
\end{theorem}

\begin{proof}
    The average
    \begin{equation}
    \begin{aligned}
        \int \dd\psi E(\ketbra{\psi}) &=\int \dd\psi \tr[ \UCh^{\otimes 2} (\rho_\psi \otimes \rho_\psi)O] \\
        &= \tr[\int \dd\psi \UCh^{\otimes 2}
        (\rho_\psi \otimes \rho_\psi)O] \\
        &= \tr[\UCh^{\otimes 2}
        \qty(\int \dd\psi \rho_\psi \otimes \rho_\psi)O] \ ,
    \end{aligned}
    \end{equation}
    where $\rho_\psi = \ketbra{\psi}$ is distributed according to the Haar measure.
    As the input state $\ket{\psi}\ket{\psi}$ is symmetric at all instances, the Haar average of the input state is a projector state that
    \begin{equation}
        \tau:= \int \dd\psi \rho_\psi \otimes \rho_\psi = \frac{\Pi_S}{\tr[\Pi_S]} \ ,
    \end{equation}
    where $\Pi_S$ is invariant under swap~\cite{watrous2018}.
    To prepare the projector state $\tau$, it is sufficient to find a basis of the projected space and then uniformly random sample the basis states.
    Thus the preparation of state $\Pi_S/\tr[\Pi_S]$ is reduced to uniformly random sample $\pmb{z}_1$ and $\pmb{z}_2$ such that $\pmb{z}_1\cdot\pmb{z}_2$ is even, and according to each pair of $(\pmb{z}_1,\pmb{z}_2)$, prepare state $\ket{\Psi_{\pmb{z}_1,\pmb{z}_2}}$.
    And the states $\ket{\Psi_{\pmb{z}_1,\pmb{z}_2}}$ can be prepared by apply the circuit~\ref{fig:sym_prepare} to classical states $\ket{\pmb{z}_1}\ket{\pmb{z}_2}$ pairwisely.
    Similar to the Bell measurement, the preparation of Bell states only requires two layers of quantum gates.

    \begin{figure}[htp]
    \centering
    \resizebox{0.35\textwidth}{!}{
    \begin{quantikz}
        \lstick{$\ket{z_1}$}  &  \qw      &  \targ{}   & \qw \\
        \lstick{$\ket{z_2}$}  &  \gate{H} &  \ctrl{-1} & \qw \\
    \end{quantikz}
    }
    \caption{The circuit that prepares Bell states.}\label{fig:sym_prepare}
    \end{figure}
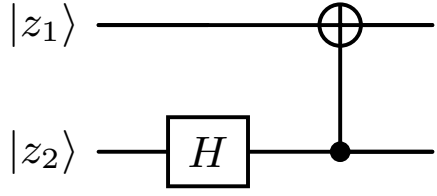

    Although some subtle optimization may be used to sample the desired classical bits $\pmb{z}_1$ and $\pmb{z}_2$, a na\"{i}ve rejective sampling can already uniformly sample the even bits in time linear to the number of qubits.
    This can be achieved by uniformly random sampling all bits in $\pmb{z}_1$ and $\pmb{z}_2$ independently.
    Then resample the bits until $\pmb{z}_1\cdot\pmb{z}_2$ is even.
    The reject rate is less than $1/2$ for all sizes of quantum circuits, as there $2^{2n-1} + 2^{n-1}$ pairs of $\pmb{z}_1$ and $\pmb{z}_2$ that $\pmb{z}_1\cdot\pmb{z}_2$ is even, and the total number of the configurations is $2^{2n}$.
\end{proof}

Combine theorem~\ref{thm:measurement} and \ref{thm:state_prep}, we have the average linear entropy $L$ of the output of a channel $\UCh$:
\begin{equation}
    \begin{aligned}
    \int \dd\psi L(\UCh(\ketbra{\psi})) =&\int \dd\psi (1 - \tr[\UCh(\ketbra{\psi})^2]) \\
    =& 1 - \int \dd\psi \tr[\UCh^{\otimes 2} (\rho_\psi \otimes \rho_\psi)S] \\
    =& 1 - \tr[\UCh^{\otimes 2} (\tau) S] \ .
    \end{aligned}
\end{equation}

This argument can be applied to the case where $L$ is defined for a subsystem.
Given the total space $\hilb$ is decomposed into $\hilb^A\otimes\hilb^B$, the average linear entropy of the output of a channel $\UCh$ on subsystem $A$ is given by
\begin{equation}
    \begin{aligned}
    \iint \dd\psi_A \dd\phi_B L_A :=&\iint \dd\psi_A \dd\phi_B L(\tr_B[\UCh(\ketbra{\psi}_A \otimes\ketbra{\varphi}_B)]) \\
    =& 1 - \iint \dd\psi_A \dd\phi_B \tr[\UCh^{\otimes 2} ({\rho_\psi}_A \otimes {\rho_\psi}_A \otimes {\rho_\phi}_B \otimes {\rho_\phi}_B)S_A\otimes \mathbb{I}_B] \\
    =& 1 - \tr[\UCh^{\otimes 2} (\tau_A\otimes \tau_B) S_A \otimes \mathbb{I}_B] \\
    =& 1 - \tr_A[ \tr_B[\UCh^{\otimes 2} (\tau_A\otimes \tau_B)] S_A] \ .
    \end{aligned}
\end{equation}

As the purity $\Tr[\rho^2]$ is obtained from classical post-processing, the purity of each subsystem can be obtained simultaneously, with the same set of measured data.
That is, for any partition of the qubits encoding $\rho$, all of the values of $\Tr[\pqty{\rho^A}^2]$ and $\Tr[\pqty{\rho^B}^2]$ together with the $\Tr[\rho^2]$ can be obtained with one shot of measurements on the quantum circuit.
This property means that our cost function
\begin{equation}
    \begin{aligned}
    &\csep = \frac{d_A^2}{d_A^2-1} \iint \dd\psi_A \dd\psi_B \frac{L_A+L_B}{2} \\
    \propto& 1 - \frac{\tr[\UCh^{\otimes 2} (\tau_A\otimes \tau_B) S_A \otimes \mathbb{I}_B] +
    \tr[\UCh^{\otimes 2} (\tau_A\otimes \tau_B)\mathbb{I}_A \otimes S_B]}{2} \\
    =& 1 - \frac{\tr[\UCh^{\otimes 2} (\tau_A\otimes \tau_B) (S_A \otimes \mathbb{I}_B + \mathbb{I}_A \otimes S_B)]}{2}
    \end{aligned}
\end{equation}
can be evaluated with a single shot of measurements on the quantum circuit, and a linear time classical post-processing.

\section{\apptwo} 
\label{sec:ParameterShift}

\subsection{Preliminary}

The parameter shift rule~\cite{mitarai2018} is widely used in variational quantum circuits for gradient-based optimization.
It outperforms the finite difference method, which estimate $\pdv{\theta_i} f(\vb*{\theta}) \approx \frac{f(\vb*{\theta} + \Delta \theta_i) - f(\vb*{\theta})}{\Delta \theta_i}$, with a finite small number $\Delta \theta_i$.
In the finite difference method, as the denominator $\Delta\theta_i$ is small, the difference of $f$ is a small number, which is relatively sensitive to the noises.
But, with the parameter shift rule, the gradient can be estimated directly, which is less sensitive to perturbations as long as the gradient is non-varnishing.

The usual parameter shift rule is based on the idempotent properties of Hermitian generators.
Let us consider a single-parameterized unitary $U(\theta_i) = e^{-i \frac{\theta_i}{2} H}$ that is generated by an idempotent Hermitian operator $H$ that $H^2 = \mathbb{I}$.
The parameter shift rule says

\begin{proposition}\label{prop:ParameterShift}
    For all observables $A$ and all initial states $\rho$, the derivative of the expectation value
    \begin{equation}
        \begin{aligned}
        \pdv{\theta_i} \tr[U(\theta_i)\rho U^\dagger(\theta_i)A]
        =& \frac{\tr[U(\theta_i+\frac{\pi}{2}) \rho U^\dagger(\theta_i+\frac{\pi}{2})A]
        - \tr[U(\theta_i-\frac{\pi}{2}) \rho U^\dagger(\theta_i-\frac{\pi}{2})A]}{2} \ .
        \end{aligned}
    \end{equation}
\end{proposition}

\begin{proof}
The derivatives of the unitary operator and its conjugate are given by $\pdv{\theta_i} U(\theta_i) = -i U(\theta_i)H/2$, and $\pdv{\theta_i} U^\dagger(\theta_i) = i HU^\dagger(\theta_i)/2$.
As a result, for any operator $\rho$,
\begin{equation}
    \begin{aligned}
        \pdv{\theta_i} U(\theta_i) \rho U^\dagger(\theta_i)
        =& \qty(\pdv{\theta_i}U(\theta_i)) \rho U^\dagger(\theta_i) + U(\theta_i) \rho \qty(\pdv{\theta}U^\dagger(\theta_i)) \\
        =& -iU(\theta_i)H \rho U^\dagger(\theta_i)/2 + iU(\theta_i) \rho H U^\dagger(\theta_i)/2 \\
        =& U(\theta_i) (-i[H,\rho]) U^\dagger(\theta_i)/2 \\
    \end{aligned}
\end{equation}
By matching the coefficients of the Taylor series, it can be shown that $U(\theta_i) = \cos(\theta_i/2)\mathbb{I} - i\sin(\theta_i/2) H$, and
\begin{equation}
\begin{aligned}
    -i[H,\rho] =& -iH\rho + i\rho H\\
    =& \frac{(\mathbb{I} - iH)\rho (\mathbb{I} + iH) - (\mathbb{I} +iH)\rho (\mathbb{I}-iH)}{2} \\
    =& e^{-\frac{i\pi}{4}H} \rho e^{\frac{i\pi}{4}H} - e^{\frac{i\pi}{4}H} \rho e^{-\frac{i\pi}{4}H} \\
    =& U(\frac{\pi}{2})\rho U^\dagger(\frac{\pi}{2}) - U(-\frac{\pi}{2})\rho U^\dagger(-\frac{\pi}{2}) \ .
\end{aligned}
\end{equation}
In summary, the derivative of the operator
\begin{equation}\label{eq:ParameterShiftRule}
    \begin{aligned}
        \pdv{\theta_i} U(\theta_i) \rho U^\dagger(\theta_i)
        =& \frac{U(\theta_i+\frac{\pi}{2}) \rho U^\dagger(\theta_i+\frac{\pi}{2})
        - U(\theta_i-\frac{\pi}{2}) \rho U^\dagger(\theta_i-\frac{\pi}{2})}{2} \\
    \end{aligned}
\end{equation}
The proposition is the direct result of \eqref{eq:ParameterShiftRule}
\end{proof}

\subsection{Parameter Shift Rule for $\csep$}
Here we show that a similar parameter shift rule can be derived for our cost function $\csep$.
Specifically, the derivative $\pdv{\theta_i} \csep$ can be evaluated by four terms, each of which can be evaluated by the same quantum circuit for evaluating $\csep$.
To simplify the notation in the following discussion, let $\UCh_{\theta_i}$ denote the family of unitary channels that $\UCh_{\theta_i}(\cdot) = U(\theta_i) \cdot U^\dagger(\theta_i)$.
The parameter shift rule~\eqref{eq:ParameterShiftRule} can be rewritten as
\begin{equation}
    \pdv{\theta_i} \UCh_{\theta_i}(\rho) = \frac{\UCh_{\theta_i+\frac{\pi}{2}}(\rho)-\UCh_{\theta_i-\frac{\pi}{2}}(\rho)}{2} \ .
\end{equation}

Our cost function is defined by the expectation value of an observable $O$
\begin{equation}
    \csep(W) = \tr[W^{\otimes 2} \rho {W^\dagger}^{\otimes 2} S] \ ,
\end{equation}
where the initial state $\rho = \tau_A \otimes \tau_B$, and $W$ is a quantum circuit parameterized by $\va*{\theta}$.
In the design of the ansatz, for each parameter $\theta_i$, there is only one Pauli rotation gate that depends on $\theta_i$ in $W(\va*{\theta})$.
We may split the circuit $W$ into three parts as in figure~\ref{fig:SplitUnitary}, where $W=PU(\theta_i)Q$.
All gates before $U(\theta_i)$ are collected in $Q$, and all gates after $U(\theta_i)$ are collected in $P$.
Then
\begin{equation}
    \begin{aligned}
    W^{\otimes 2} =& (PU(\theta_i)Q)^{\otimes 2} \\
    =& P^{\otimes 2} U(\theta_i)^{\otimes 2} Q^{\otimes 2} \\
    =& P^{\otimes 2} (U(\theta_i)\otimes \mathbb{I})(\mathbb{I}\otimes U(\theta_i)) Q^{\otimes 2} \ .
    \end{aligned}
\end{equation}
Let $\UCh^\prime_{\theta_i}$ and $\UCh^{''}_{\theta_i}$ denote the unitary channels corresponding to $U(\theta_i)\otimes \mathbb{I}$ and $\mathbb{I}\otimes U(\theta_i)$ respectively.
And let $\mathcal{P}, \mathcal{Q}$ denote the unitary channels corresponding to $P^{\otimes 2}$ and $Q^{\otimes 2}$.
Then the unitary evolution can be written as the composition of $4$ unitary channels:
\begin{equation}
    \csep(W) = \tr[(\mathcal{P}\circ \UCh^{'}_{\theta_i}\circ \UCh^{''}_{\theta_i}\circ \mathcal{Q}) (\rho) O] \ .
\end{equation}

\begin{figure}[htp]
    \centering
    \includegraphics[width=.45\textwidth]{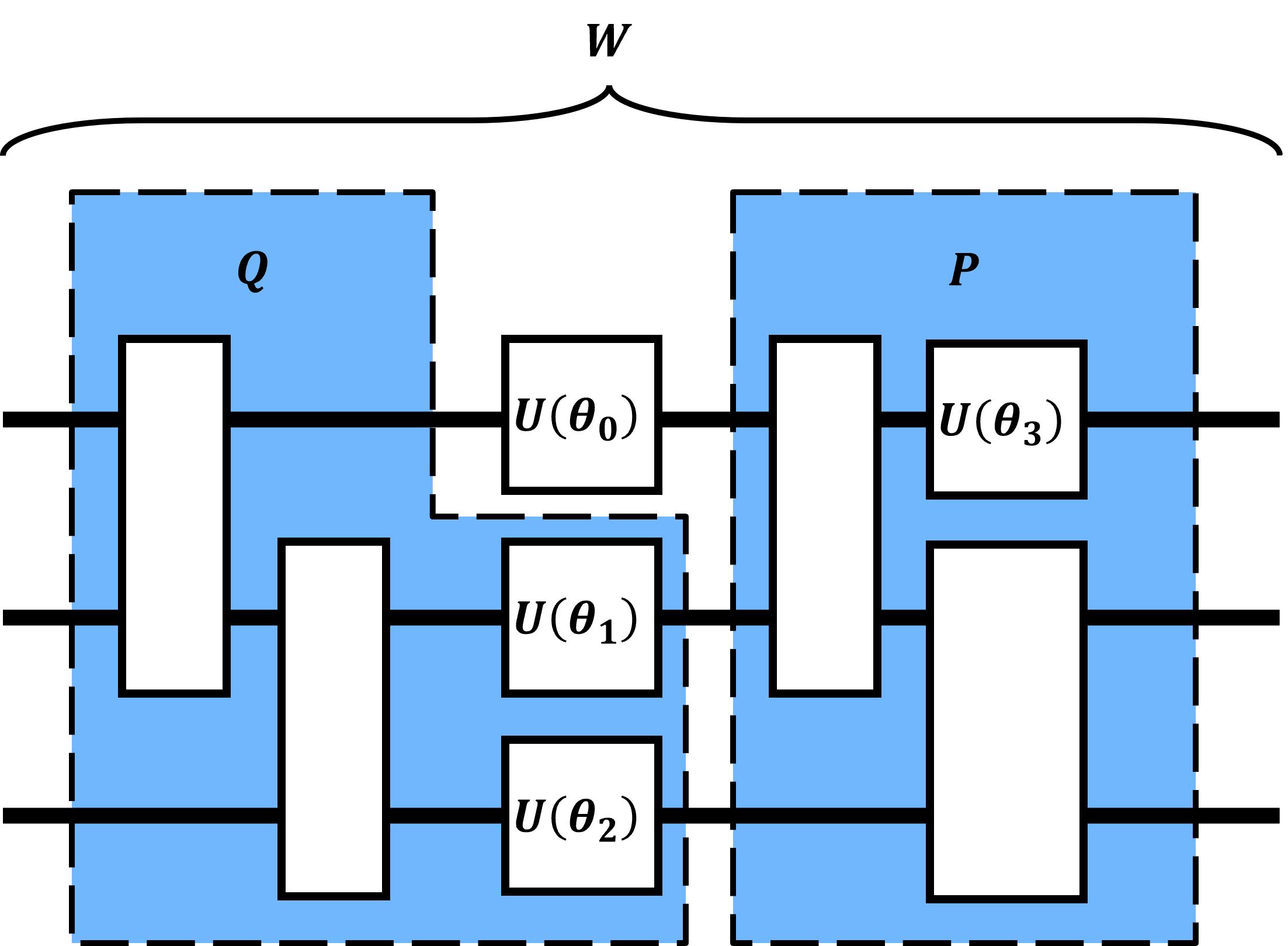}
    \caption{For each parameter $\theta_i$, the parameterized circuit $W$ can be split into three parts, where $W=PU(\theta_i)Q$.
    As there is only one Pauli rotation gate $U(\theta_i)$ that depends on $\theta_i$, all gates before $U(\theta_i)$ are collected in $Q$, and all gates after $U(\theta_i)$ are collected in $P$.
    For convenience, we put all gates parallel to $U(\theta_i)$ into $Q$.}
    \label{fig:SplitUnitary}
\end{figure}

According to the product rule of derivatives,
\begin{equation}
    \pdv{\theta_i} (\UCh^\prime_{\theta_i}\circ \UCh^{''}_{\theta_i})
    = (\pdv{\theta_i} \UCh^\prime_{\theta_i})\circ \UCh^{''}_{\theta_i}
    + \UCh^\prime_{\theta_i} \circ (\pdv{\theta_i} \UCh^{''}_{\theta_i}) \ .
\end{equation}
As the Paul rotations $U(\theta_i)$ are generated by idempotent operators and satisfies the parameter shift rule~\eqref{eq:ParameterShiftRule}, the same property holds for $\UCh^\prime_{\theta_i}$ and $\UCh^{''}_{\theta_i}$.
In conclusion, the derivative of our cost function is
\begin{equation}
    \begin{aligned}
    \pdv{\theta_i} \csep(W) =& \tr[(\mathcal{P}\circ \pdv{\theta_i} \qty(\UCh^{'}_{\theta_i}\circ \UCh^{''}_{\theta_i})\circ \mathcal{Q}) (\rho) O] \\
    =& \tr[(\mathcal{P}\circ \qty(\frac{\UCh^{'}_{\theta_i+\frac{\pi}{2}}-\UCh^{'}_{\theta_i-\frac{\pi}{2}}}{2})\circ \UCh^{''}_{\theta_i}\circ \mathcal{Q}) (\rho) O]
    + \tr[(\mathcal{P}\circ \UCh^{\prime}_{\theta_i}\circ \qty(\frac{\UCh^{''}_{\theta_i+\frac{\pi}{2}}-\UCh^{''}_{\theta_i-\frac{\pi}{2}}}{2})\circ \mathcal{Q}) (\rho) O] \\
    =& \frac{1}{2}\tr[(\mathcal{P}\circ \UCh^{'}_{\theta_i+\frac{\pi}{2}}\circ \UCh^{''}_{\theta_i}\circ \mathcal{Q}) (\rho) O]
    - \frac{1}{2}\tr[(\mathcal{P}\circ \UCh^{'}_{\theta_i-\frac{\pi}{2}}\circ \UCh^{''}_{\theta_i}\circ \mathcal{Q}) (\rho) O] \\
    &+ \frac{1}{2}\tr[(\mathcal{P}\circ \UCh^{\prime}_{\theta_i}\circ \UCh^{''}_{\theta_i+\frac{\pi}{2}}\circ \mathcal{Q}) (\rho) O]
    - \frac{1}{2}\tr[(\mathcal{P}\circ \UCh^{\prime}_{\theta_i}\circ \UCh^{''}_{\theta_i-\frac{\pi}{2}}\circ \mathcal{Q}) (\rho) O] \ .
    \end{aligned}
\end{equation}
This is the sum of $4$ terms, each of which can be evaluated by the same method of evaluating $\csep(W)$ with a ``parameter shift''.

\section{Relating decoupling cost to the gate fidelity}

\subsection{Characteristics of the gate fidelity}

The closeness of approximating $U$ with $U_A\otimes U_B$ is measured by the gate fidelity between $U$ and $U_A\otimes U_B$, which is defined by
\begin{equation}
    \bar{F} := \int \dd\psi \tr[(U_A\otimes U_B)^\dagger U \ketbra{\psi}{\psi} U^\dagger (U_A\otimes U_B)\ketbra{\psi}{\psi}] \ .
\end{equation}
This Haar average can be contracted to a single quantity that
\begin{equation}
    \bar{F} = \frac{1}{d+1} + \frac{1}{d(d+1)} \abs{\langle U_A\otimes U_B, U\rangle}^2 \ ,
\end{equation}
where $d$ is the dimension of the global Hilbert space.
Here $\langle U_A\otimes U_B, U\rangle := \Tr[(U_A\otimes U_B)^\dagger U]$ denotes the Hilbert-Schmidt inner product between $U_A\otimes U_B$ and $U$.
Note that, the average fidelity is maximized when $\langle U_A\otimes U_B, U\rangle$ is maximized.
Let $\UCh$ be the unitary channel that $\UCh(\cdot) = U(\cdot)U^\dagger$, and let $\UCh_A$ and $\UCh_B$ be the corresponding channels for $U_A$ and $U_B$.
Denote the Choi-Jamiolkowski matrix of a channel $\UCh$ by $J^\UCh$.
Then the inner product can be expressed as
\begin{equation}
    \abs{\langle U_A\otimes U_B, U\rangle}^2 = \Tr[(J^{\UCh_A}\otimes J^{\UCh_B}) J^\UCh] \ .
\end{equation}
As the Choi rank of unitary channels are all $1$, we may rewrite the Choi operators as the outer products
\begin{equation}
    J^\UCh = \ketbra{\psi}{\psi}, \quad J^{\UCh_A} = \ketbra{\psi_A}{\psi_A}, \quad J^{\UCh_B} = \ketbra{\psi_B}{\psi_B} \ .
\end{equation}
Then
\begin{equation}
    \abs{\langle U_A\otimes U_B, U\rangle}^2 = \abs{(\bra{\psi_A}\otimes \bra{\psi_B})\ket{\psi}}^2 \ .
\end{equation}
As $\ket{\psi}$ is a vector in the Hilbert space, there exists a Schmidt decomposition such that
\begin{equation}\label{eq:SchmidtFid}
    \ket{\psi} = \sum_k \sqrt{p_k} \ket{\xi_k}\otimes \ket{\phi_k} \ ,
\end{equation}
where $\sum_k p_k = \braket{\psi} = d$ according to the normalization condition for quantum channels.
As the result, the maximum value of $\abs{\langle U_A\otimes U_B, U\rangle}^2$ is given by
\begin{equation}
    \max_{U_A,U_B} \abs{\langle U_A\otimes U_B, U\rangle}^2 = \max_{\ket{\psi_A},\ket{\psi_B}} \abs{\sum_k \sqrt{p_k} \braket{\psi_A}{\xi_k}\otimes \braket{\psi_B}{\phi_k}}^2 \ .
\end{equation}
Let $q_{A,k}:= \braket{\psi_A}{\xi_k}$ and $q_{B,k}:= \braket{\psi_B}{\phi_k}$, we have
\begin{equation}
\begin{aligned}
    \abs{\langle U_A\otimes U_B, U\rangle} =& \abs{\sum_k \sqrt{p_k} q_{A,k} q_{B,k}} \\
    \leq& \sum_k \sqrt{p_i}\abs{q_{A,k} q_{B,k}} \\
    \leq& \max{\sqrt{p_i}} \sum_k \abs{q_{A,k} q_{B,k}} \\
    =& \max{\sqrt{p_i}} \abs{\braket{\abs{p_A}}{\abs{q_B}}} \\
\end{aligned}
\end{equation}
According to the Cauchy-Schwartz inequality,
\begin{equation}
    \abs{\braket{\abs{q_A}}{\abs{q_B}}}^2 \leq \norm{q_A}_2^2 \norm{q_B}_2^2 \ .
\end{equation}
Given $\ket{\xi_k}$ and $\ket{\phi_k}$ are orthonormal bases, $\norm{q_A}_2^2 = \abs{\braket{\psi_A}}^2 = d_A$ and $\norm{q_B}_2^2=d_B$.
As above inequalities hold for all $\ket{\psi_A}, \ket{\psi_B}$,
\begin{equation}
    \max_{U_A,U_B} \abs{\langle U_A\otimes U_B, U\rangle}^2 \leq \max{q_i} d_A d_B = d \norm{p}_\infty \ .
\end{equation}
where the infinity norm $\norm{p}_\infty$ equals to the maximum Schmidt coefficient of $\ket{\psi}$.
The equality holds when $\ket{\psi_A}$ and $\ket{\psi_B}$ are proportional to the $\ket{\xi_k}$ and $\ket{\phi_k}$ that corresponding to the maximum $p_k$.
However, this might not be achievable as $\ket{\xi_k}$ and $\ket{\phi_k}$ do not necessarily correspond to unitary operators.

In summary, for all $U_A\otimes U_B$, the gate fidelity is bounded by
\begin{equation}
    \frac{1}{d+1} \leq \bar{F}(U_A\otimes U_B, U) \leq \frac{1}{d+1} (\norm{p}_\infty + 1) \leq 1 \ .
\end{equation}
In other words
\begin{equation}
    \bar{F}_{\max}:= \max_{U_A,U_B} \bar{F}(U_A\otimes U_B, U) \leq \frac{1}{d+1} (\norm{p}_\infty + 1)
\end{equation}

To fully understand the property of separability, we would like to explore the gate fidelity with respect to an alternative set.
Without loss of generality, assuming the dimension of the system $A$, $d_A$, is not larger than $d_B$.
Then we may extend the Hilbert space of $A$ to $A^\prime$, whose dimension is the same as $B$, such that $A^\prime = A\oplus C$, where $C$ is a $d_B-d_A$ dimensional Hilbert space.
We may extend $U$ to construct a unitary operator in $A^+$, such that $U^\prime = U\oplus \mathbb{I}_{B\otimes C}$.
As the dimensions of $A^+$ and $B$ are the same, we may define the swap operator $S$ between $A^+$ and $B$.
Now, we will focus on the gate fidelity
\begin{equation}\label{eq:extendedU}
    \max_{U_{A^+},U_B} \bar{F}(U_{A^+}\otimes U_B, SU^\prime) \ .
\end{equation}
Let $\qty{q_k}$ be the Schmidt coefficients of the Choi operator with respect to $SU^\prime$.
Then, with similar arguments as above, we have
\begin{equation}
    \bar{F}_{\max}^S:= \max_{U_{A^+},U_B} \bar{F}(U_{A^+}\otimes U_B, SU^\prime) \leq \frac{1}{d_B^2+1} (\norm{q}_\infty +1)  \ .
\end{equation}
Since $d_B^2 \geq d_A d_B = d$ by assumption,
\begin{equation}
    \bar{F}_{\max}^S \leq \frac{1}{d+1} (\norm{q}_\infty +1)  \ .
\end{equation}
This quantity would be useful in the analysis of the decoupling cost.

\subsection{Characteristics of the decoupling cost}

In this section, we show that the maximum gate fidelity $\bar{F}_{\max}$ is related to the decoupling cost function $\csep$, such that
\begin{equation}
    \bar{F}_{\max} \leq 1- \csep + \frac{3}{d+1} \ .
\end{equation}

The evaluation of $\csep$ requires two copies of a target unitary $U$, whose acts on the system $A\otimes B$.
We may label the systems as $A_1,A_2,B_1,B_2$, where one unitary acts on $A_1\otimes B_1$, and the other acts on $A_2 \otimes B_2$.
In terms of the Choi-Jamiolkowski operator, the average linear entropy, denoted $\bar{L}$, is given by
\begin{equation}
    \begin{aligned}
    &1 - \frac{\tr[\UCh^{\otimes 2} (\tau_A\otimes \tau_B) S_A \otimes \mathbb{I}_B] + 
    \tr[\UCh^{\otimes 2} (\tau_A\otimes \tau_B)\mathbb{I}_A \otimes S_B]}{2} \\
    =& 1 - \frac{\tr[\tau_A \otimes \tau_B\otimes S_{A^\prime} \otimes \mathbb{I}_{B^\prime} (J^{\UCh}\otimes J^{\UCh})]}{2}
    - \frac{\tr[\tau_A \otimes \tau_B\otimes \mathbb{I}_{A^\prime} \otimes S_{B^\prime} (J^{\UCh}\otimes J^{\UCh})]}{2} \ ,
    \end{aligned}
\end{equation}
where $S_A$ ($S_B$) is the swap operator between the output systems $A_1$ and $A_2$, (or $B$ systems respectively), and the identity operators $\mathbb{I}_A, \mathbb{I}_B$ are also defined on the system $A_1 \otimes A_2$ or $B_1\otimes B_2$ respectively.
This equality comes from theorem~\ref{thm:Choi}, where the partial transposes from the formula are omitted, as $\tau_k, \mathbb{I}_k$ and $S_k$ are all invariant under the transpose.
In this Choi representation, we need to distinguish the output systems and input systems, and thus there are $8$ subsystems in the formula in total.
We may label the subsystems as $A_1,A_2,B_1,B_2$ for the input systems, and $A_1^\prime,A_2^\prime,B_1^\prime,B_2^\prime$ for the output systems.
Note that, the initial state $\tau_k = \frac{P_s}{\tr[P_s]} = \frac{1}{2\tr[P_s]}(S_k + \mathbb{I})$, and the trace of the projector $tr[P_S]$ is $\frac{d_A^2+d_A}{2}$ for the pair of systems $(A_1, A_2)$ and $\frac{d_B^2+d_B}{2}$ for the pair of systems $(B_1, B_2)$, where $d_A,d_B$ are the dimension of system $A_i, B_i$ respectively.
Then each of the two terms in $2(1-\bar{L})$ is divided into $4$ terms:
\begin{equation}\label{eq:8terms}
\begin{aligned}
    &\tr[\tau_A \otimes \tau_B\otimes S_{A^\prime} \otimes \mathbb{I}_{B^\prime} (J^{\UCh}\otimes J^{\UCh})]
    + \tr[\tau_A \otimes \tau_B\otimes \mathbb{I}_{A^\prime} \otimes S_{B^\prime} (J^{\UCh}\otimes J^{\UCh})] \\
    =& \frac{1}{(d_A^2+d_A)(d_B^2+d_B)}\Big(\tr[\mathbb{I}_A \otimes \mathbb{I}_B \otimes S_{A^\prime} \otimes \mathbb{I}_{B^\prime} (J^{\UCh}\otimes J^{\UCh})] \\
    &+ \tr[S_A \otimes \mathbb{I}_B \otimes S_{A^\prime} \otimes \mathbb{I}_{B^\prime} (J^{\UCh}\otimes J^{\UCh})]
    + \tr[\mathbb{I}_A \otimes S_B \otimes S_{A^\prime} \otimes \mathbb{I}_{B^\prime} (J^{\UCh}\otimes J^{\UCh})] \\
    &+ \tr[S_A \otimes S_B \otimes S_{A^\prime} \otimes \mathbb{I}_{B^\prime} (J^{\UCh}\otimes J^{\UCh})]
    +\tr[\mathbb{I}_A \otimes \mathbb{I}_B \otimes \mathbb{I}_{A^\prime}\otimes S_{B^\prime}  (J^{\UCh}\otimes J^{\UCh})] \\
    &+ \tr[S_A \otimes \mathbb{I}_B \otimes \mathbb{I}_{A^\prime}\otimes S_{B^\prime}  (J^{\UCh}\otimes J^{\UCh})]
    + \tr[\mathbb{I}_A \otimes S_B \otimes \mathbb{I}_{A^\prime}\otimes S_{B^\prime}  (J^{\UCh}\otimes J^{\UCh})] \\
    &+ \tr[S_A \otimes S_B \otimes \mathbb{I}_{A^\prime}\otimes S_{B^\prime}  (J^{\UCh}\otimes J^{\UCh})]\Big) \ .
\end{aligned}
\end{equation}
For convenience, we label each term in the parenthesis as $P_k$ in the order presented in \eqref{eq:8terms}, where $k$ ranges from $1$ to $8$.
That is
\begin{equation}
    2(1-\bar{L}) = \frac{\sum_k P_k}{(d_A^2+d_A)(d_B^2+d_B)} \ .
\end{equation}%
By changing the order of the tensor products and grouping the swap operators, we may rewrite each $P_k$ in the form of $\tr[S_\alpha \otimes \mathbb{I}_\beta (J_{\alpha_1,\beta_1}\otimes J_{\alpha_2,\beta_2})]$ for some partitioning $(\alpha,\beta)$ of the four pairs the systems $(A_1,A_2)$, $(B_1,B_2)$, $(A^\prime_1,A^\prime_2)$, $(B^\prime_1,B^\prime_2)$.
This formula is equivalent to
\begin{equation}
    \tr[S_\alpha \tr_{\beta}[J_{\alpha_1,\beta_1}]\otimes \tr_{\beta}[J_{\alpha_2,\beta_2}]] \ .
\end{equation}
Given $J^{\UCh}$ is a rank-$1$ operator, we can write it in its Schmidt decomposition $\ket{\psi} = \sum_k \sqrt{c_k}\ket{\varphi_k}_\alpha\otimes \ket{\varepsilon_k}_\beta$.
After the partial trace of the $\beta$ systems,
\begin{equation}
    \tr[J_{\alpha_1,\beta_1}]\otimes \tr[J_{\alpha_2,\beta_2}]
    = \qty(\sum_k c_k \ketbra{\varphi_k}{\varphi_k})\otimes \qty(\sum_l c_l \ketbra{\varphi_l}{\varphi_l}) \ .
\end{equation}
Then, by applying the swap and the trace
\begin{equation}
    \tr[S_\alpha \tr_{\beta}[J\otimes J]] =
    \sum_{k,l} c_k c_l \braket{\varphi_l}{\varphi_k} \braket{\varphi_k}{\varphi_l}
    = \norm{c}_2^2 \ .
\end{equation}
This result indicates that each term of the $8$ terms is the 2-norm of the Schmidt coefficients of $\ket{\psi}$ for some partitioning $(\alpha,\beta)$ of the system $J^\UCh$ acting on.
That is saying, every term $P_k$ is positive and upper bounded by the inner product of $\ket{\psi}$, which is $d$.

Now, we need to analyze each of the terms individually.
First notice that each of the terms is the $2$-norm of the Schmidt coefficients of $\ket{\psi}$, which only depends on the partitioning of the system.
Specifically,
\begin{equation}
    \tr[S_\alpha \tr_{\beta}[J\otimes J]] = \norm{c}_2^2 = \tr[S_\beta \tr_{\alpha}[J\otimes J]] \ .
\end{equation}
As the result, by interchanging the subsystems $\alpha$ and $\beta$, the $8$ terms can be grouped into $4$ terms:
\begin{equation}
    P_1 = P_8;\quad P_2 = P_7;\quad P_3 = P_6;\quad P_4 = P_5 \ .
\end{equation}

$P_2$ is strongly related to the gate fidelity between $U$ and $U_A\otimes U_B$.
For $P_2$, the partition $(\alpha,\beta)$ matches the partition $(A,B)$ in the analysis of the fidelity as in \eqref{eq:SchmidtFid}, which means $P_2 = \norm{p}_2^2$.

$P_1$ and $P_5$ are equivalent to the expectation value of $S_A\otimes \mathbb{I}_B$ and $\mathbb{I}_A\otimes S_B$ respectively with completely mixed state as input.
That is
\begin{equation}
    \begin{aligned}
    P_1 &= \tr[(S_A\otimes \mathbb{I}_B) \UCh(\mathbb{I})\otimes\UCh(\mathbb{I})] \\
    &= \tr[(S_A\otimes \mathbb{I}_B) \mathbb{I}\otimes\mathbb{I}] \\
    &= \tr[(S_A\otimes \mathbb{I}_B)] \\
    &= d_A d_B^2 \ ,
    \end{aligned}
\end{equation}
where $d_A$ and $d_B$ are the dimensions of the subsystems $A_i$ and $B_i$ respectively.
Similarly, $P_5 = d_A^2 d_B$.

In summary,
\begin{equation}
    \begin{aligned}
    & 1-\bar{L} =&\frac{\norm{p}_2^2 + d_Ad_B^2 + d_Bd_A^2 + P_3}{(d_A^2+d_A)(d_B^2+d_B)} \ .
    \end{aligned}
\end{equation}
As the dimension of the whole system $A\otimes B$ is $d = d_A d_B$, we have
\begin{equation}
    \bar{L} = -\frac{\norm{p}_2^2 - d^2 + P_3 -d}{d(d+1+d_A+d_B)} \ .
\end{equation}
Which means
\begin{equation}\label{eq:p2norm}
    \norm{p}_2^2 = d^2 - P_3 +d -d(d+1+d_A+d_B)\bar{L} \ .
\end{equation}

Now, we will substitute the definition $\csep = \frac{d_A^2}{d_A^2-1}\bar{L}$ to \eqref{eq:p2norm}, where we assume $d_A \leq d_B$.
Let $K:= \frac{d_A^2-1}{d_A^2}$, then $\csep = \frac{\bar{L}}{K}$.
\begin{equation}
    \norm{p}_2^2 = d^2+d- P_3 - K\csep d(d+1+d_A+d_B) \ .
\end{equation}

Compared to the gate fidelity that
\begin{equation}
    \begin{aligned}
    \bar{F}_{\max} \leq& \frac{1}{d+1} + \frac{1}{d+1}\norm{p}_\infty \\
    \leq&\frac{1}{d+1} + \frac{1}{d+1}\norm{p}_2 \\
    = & \frac{1}{d+1} + \frac{1}{d+1}\sqrt{d^2+d-P_3 - K\csep d(d+1+d_A+d_B)} \\
    \leq & \frac{1}{d+1} + \frac{1}{d+1}\sqrt{d^2+d-P_3 - K\csep (d^2+d)} \\
    = & \frac{1}{d+1} + \sqrt{\frac{d}{d+1}}\sqrt{1-\frac{P_3}{d^2+d} - K\csep} \ .
    \end{aligned}
\end{equation}
We may organize the terms to get
\begin{equation}
    \qty(\bar{F}_{\max} - \frac{1}{d+1})^2 \leq  \frac{d}{d+1} \qty(1 - K\csep - \frac{P_3}{d^2+d}) \ .
\end{equation}
and
\begin{equation}
    \qty(\bar{F}_{\max} - \frac{1}{d+1})^2 + \frac{P_3}{(d+1)^2} \leq \frac{d}{d+1} \qty(1 - K\csep) \ .
\end{equation}

As $d = d_A d_B \geq d_A^2$ by assumption, we have, $K \leq \frac{d^2-1}{d^2}$, and the above equation can be further simplified to
\begin{equation}\label{eq:Bound}
    \begin{aligned}
    \qty(\bar{F}_{\max} - \frac{1}{d+1})^2 + \frac{P_3}{(d+1)^2}
    \leq& \frac{d}{d+1} \qty(1 - \frac{d^2-1}{d^2}\csep ) = \frac{d}{d+1} - \frac{d-1}{d}\csep  \ .
    \end{aligned}
\end{equation}
The quantity of $P_3$ here is related to the extended unitary $U^\prime$ as defined in \eqref{eq:extendedU}.
As $J^{\UCh^\prime} = J^{\UCh}\oplus J^{\mathbb{I}_{B\otimes C}}$ and $J^{\mathbb{I}_{B\otimes C}} = \ketbra{\Psi^+}_{B\otimes C}$, we have
\begin{equation}
    \begin{aligned}
    &\tr[S_{A^+} \otimes \mathbb{I}_B \otimes \mathbb{I}_{{A^+}^\prime}\otimes S_{B^\prime} (J^{\UCh^\prime}\otimes J^{\UCh^\prime})] \\
    =&\tr[S_{A} \otimes \mathbb{I}_B \otimes \mathbb{I}_{A^\prime}\otimes S_{B^\prime} (J^{\UCh}\otimes J^{\UCh})]
    +\tr[S_{C} \otimes \mathbb{I}_B \otimes \mathbb{I}_{C^\prime}\otimes S_{B^\prime} (J^{\mathbb{I}_{B\otimes C}}\otimes J^{\mathbb{I}_{B\otimes C}})] \\
    =&\tr[S_{A} \otimes \mathbb{I}_B \otimes \mathbb{I}_{A^\prime}\otimes S_{B^\prime} (J^{\UCh}\otimes J^{\UCh})]
    + \tr[S_C]\tr[S_{B^\prime}] \\
    =&\tr[S_{A} \otimes \mathbb{I}_B \otimes \mathbb{I}_{A^\prime}\otimes S_{B^\prime} (J^{\UCh}\otimes J^{\UCh})]
    + (d_B^2 - d) \\
    =& P_3 + (d_B^2 - d) \ ,
    \end{aligned}
\end{equation}
where the first equality comes from the fact that both $S_{A^+}$ and $\mathbb{I}_{A^+}$ are only supported on $A\otimes A$ and $C\otimes C$, but not $A\otimes C$ and $C\otimes A$, and hence the $J^{\UCh}\otimes J^{\mathbb{I}_{B\otimes C}}$ and $J^{\mathbb{I}_{B\otimes C}}\otimes J^{\UCh}$ terms are zero.
Observe that, similar to the analysis of $P_2$, the $P_3$ is the 2-norm of the Schmidt coefficients
\begin{equation}
    \begin{aligned}
    &\tr[S_{A^+} \otimes \mathbb{I}_B \otimes \mathbb{I}_{{A^+}^\prime}\otimes S_{B^\prime} (J^{\UCh^\prime}\otimes J^{\UCh^\prime})] \\
    =&\tr[S_{A^+} \otimes \mathbb{I}_B \otimes S_{{A^+}^\prime}\otimes \mathbb{I}_{B^\prime} (J^{S\UCh^\prime}\otimes J^{S\UCh^\prime})] \\
    =& \norm{q}_2^2 \ ,
    \end{aligned}
\end{equation}
where the output systems ${A^+}^\prime$ and $B^\prime$ are swapped in the first equality.
In summary, we have
\begin{equation}
    \begin{aligned}
    P_3 &= \norm{q}_2^2  - (d_B^2-d) \geq \norm{q}_\infty^2 - (d_B^2 -d) \\
    &\geq (d+1)^2 (\bar{F}^S_{\max} - \frac{1}{d+1})^2 - (d_B^2 -d) \ ,
    \end{aligned}
\end{equation}
Where $\bar{F}^S_{\max}$ is the maximum gate fidelity of the extended unitary $SU^\prime$ as defined in \eqref{eq:extendedU}.
Substitute this to \eqref{eq:Bound}, we have
\begin{equation}
    \qty(\bar{F}_{\max} - \frac{1}{d+1})^2 + \qty(\bar{F}^S_{\max} - \frac{1}{d+1})^2 \leq \frac{d^2+d_B^2}{(d+1)^2} - \frac{d-1}{d}\csep  \ .
\end{equation}
In the case $d_A = d_B$, we have $d = d_B^2$, and
\begin{equation}
    \begin{aligned}
    &\qty(\bar{F}_{\max} - \frac{1}{d+1})^2 + \qty(\bar{F}^S_{\max} - \frac{1}{d+1})^2 \leq \frac{d}{d+1} - \frac{d-1}{d}\csep \leq \qty(1-\frac{1}{d})(1-\csep)  \ ,
    \end{aligned}
\end{equation}
and the extended unitary $U^\prime$ is the same as $U$.
It shows that the maximum gate fidelities $\bar{F}_{\max}$ and $\bar{F}^S_{\max}$ are jointly bounded by $1-\csep$.

To obtain a neater but looser bound shown in theorem 2 in the main text, we start from \eqref{eq:Bound}, then observe that $P_3$ is positive, and $\bar{F}_{\max} \leq 1$, such that
\begin{equation}
    \bar{F}_{\max}^2 - \frac{2}{d+1} +\qty(\frac{1}{d+1})^2 \leq \qty(\bar{F}_{\max} - \frac{1}{d+1})^2 \leq \frac{d}{d+1} - \frac{d-1}{d}\csep  \ .
\end{equation}
By combining the factors, we have
\begin{equation}
    \begin{aligned}
    \bar{F}_{\max}^2 &\leq \frac{2d+1}{(d+1)^2} + \frac{d^2+d}{(d+1)^2} - \frac{d-1}{d}\csep \\
    &\leq \frac{d^2+3d+1}{(d+1)^2} - \frac{d-1}{d}\csep \\
    &\leq \frac{d^2+3d+1}{(d+1)^2} - \frac{d-1}{d+1}\csep \\
    &= \frac{d-1}{d+1}(1 - \csep) + \frac{3d+2}{(d+1)^2}  \\
    &\leq (1 - \csep) + \frac{3}{d+1}  \ ,
    \end{aligned}
\end{equation}
which proves theorem 2 in the main text.

\section{Extended Numerical Results}

In this section, we provide the extended numerical results for the main text.
In Fig.~\ref{fig:4qubit_sturcture}, we show that, if the 4-qubit unitary is generated by a circuit with a similar structure as our ansatz, our method can achieve a higher fidelity.

\begin{figure}
\centering
\includegraphics[width=0.5\columnwidth]{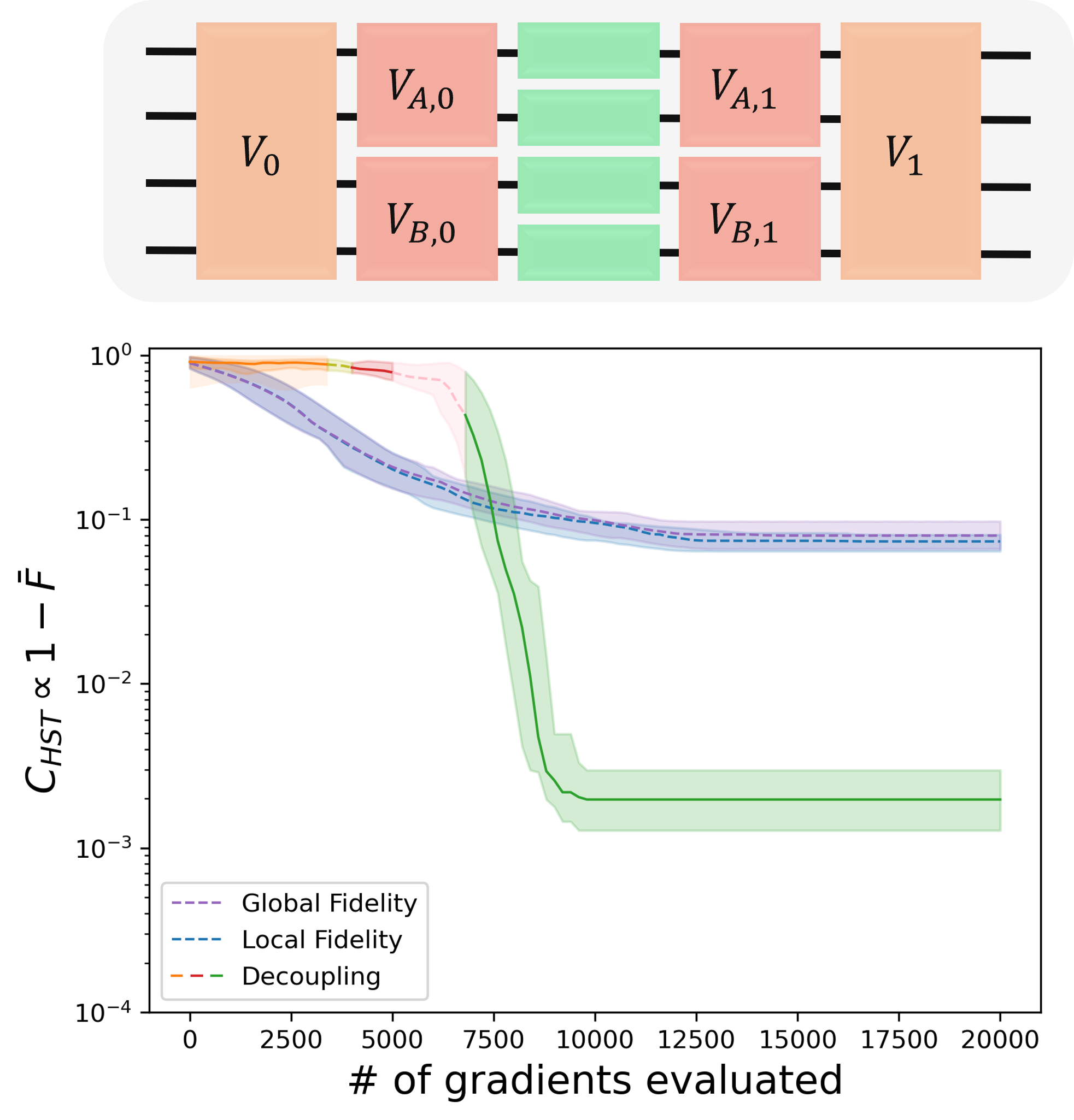}
\caption{\emph{Variational Compilation of 4-qubit unitary that is generated by a spindle-shaped circuit.}
In the main text, we compared our method with direct quantum compilation for Haar randomly generated 4-qubit unitary.
However, as our ansatz is shallow, it is not expected to compile any 4-qubit unitary with high fidelity.
Here we show that, if the 4-qubit unitary is generated by a circuit with a similar structure as our ansatz, our method can achieve a higher fidelity.
Similar to the ansatz in the main text, the circuit is divided into 3 types: 4-qubit ansatzes (orange), 2-qubit ansatzes (red), and the middle single qubit gates (green).
However, the ansatzes $V_0, V_1$, and $V_{A,0}, V_{B,0}, V_{A,1}, V_{B,1}$ are not repeated, and each of them only contains one layer of CNOT gates.
}
\label{fig:4qubit_sturcture}
\end{figure}

For a better understanding of our numerical results, we plot the raw data of Fig.~\ref{fig:4qubit_sturcture} explicitly in Fig.~\ref{fig:4qubit_sturcture_raw}.
The solid and dashed lines are the global fidelities, whose mean and the ranges are used to generate Fig.~\ref{fig:4qubit_sturcture}.
The dotted lines are the corresponding cost functions used at each phase of the optimization.

\begin{figure}[htp]
    \centering
    \includegraphics[width=0.5\textwidth]{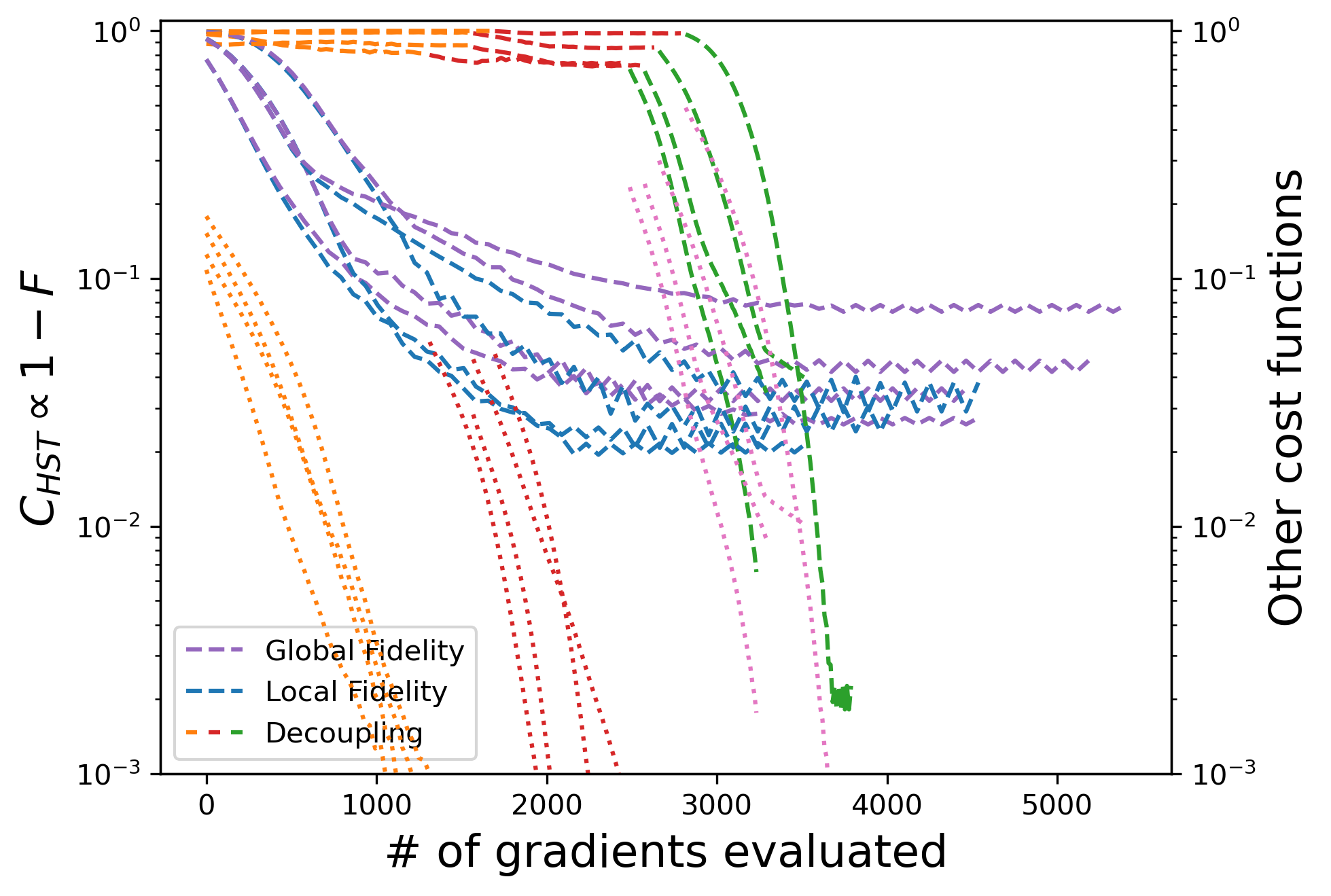}
    \caption{\emph{Raw data for the 4-qubit unitary generated by a spindle-shaped circuit.}}
    \label{fig:4qubit_sturcture_raw}
\end{figure}

\newpage

\bibliographystyle{unsrt}
\bibliography{main}

\end{document}